\documentclass[11pt]{amsart}

\makeatletter
\@namedef{subjclassname@2020}{\textup{2020} Mathematics Subject Classification}
\makeatother

\usepackage{amssymb,verbatim,graphicx}
\usepackage[mathscr]{eucal}
\usepackage{enumerate}
\usepackage{xspace}
\usepackage[thinlines]{easymat}
\usepackage{tabularx}
\usepackage{subfig}
\usepackage{hyperref}    
\usepackage{cleveref}

\newtheorem{theorem}{Theorem}[section]
\newtheorem{lemma}[theorem]{Lemma}
\newtheorem{proposition}[theorem]{Proposition}

\newtheorem{definition}[theorem]{Definition}

\theoremstyle{definition}
\newtheorem{example}[theorem]{Example}
\newtheorem{remark}[theorem]{Remark}
\newtheorem{problem}[theorem]{Problem}
\newtheorem{principle}{Guiding Principle}

\numberwithin{equation}{section}

\def\&{\wedge}

\begin{document}

\title{Mend the gap: A smart repair algorithm for noisy polygonal tilings}

\author{Jeanne N. Clelland}
\address{Department of Mathematics, 395 UCB, University of
Colorado, Boulder, CO 80309-0395}
\email{Jeanne.Clelland@colorado.edu}



\begin{abstract}
Let $T^* = \{P^*_1, \ldots, P^*_N\}$ be a polygonal tiling of a simply connected region $R^*$ in the plane, and let $T = \{P_1, \ldots, P_N\}$ be a noisy version of $T^*$ obtained by making small perturbations to the coordinates of the vertices of the polygons in $T^*$.  In general, $T$ will only be an approximate tiling of a region $R$ that closely approximates $R^*$, due to the presence of gaps and overlaps between the perturbed polygons in $T$. The areas of these gaps and overlaps are typically small relative to the areas of the polygons themselves.

Suppose that we are given the approximate tiling $T$ and we wish to recover the tiling $T^*$. To address this problem, we introduce a new algorithm, called {\tt smart\_repair},  to modify the polygons in $T$ to produce a tiling $\widetilde{T} = \{\widetilde{P}_1, \ldots, \widetilde{P}_N\}$ of $R$ that closely approximates $T^*$, with special attention given to reproducing the {\em adjacency relations} between the polygons in $T^*$ as closely as possible. 

The motivation for this algorithm comes from computational redistricting, where algorithms are used to build districts from smaller geographic units (e.g., voting precincts).  Because districts in most U.S. states are required to be contiguous, these algorithms are fundamentally based on adjacency relations between units.  Unfortunately, the best available map data for unit boundaries is often noisy, containing gaps and overlaps between units that can lead to substantial inaccuracies in the adjacency relations.  Simple repair algorithms commonly included in geographical software packages can actually exacerbate these inaccuracies, with the result that algorithmically drawn districts based on the ``repaired" units may be discontiguous, and hence not legally compliant.  The algorithm presented here is specifically designed to avoid such problems to the greatest extent possible.

A Python implementation of the {\tt smart\_repair} algorithm is publicly available as part of the MGGG Redistricting Lab's {\tt Maup} package, available at \url{https://github.com/mggg/maup}.

\end{abstract}

\maketitle

\section{Introduction}\label{intro-sec}

\subsection{Motivation: A geometry problem in computational redistricting}\label{repair-gone-wrong-subsec}
In recent years, computational algorithms have played a rapidly growing role in the analysis of political districting plans; in particular, a variety of algorithms have been developed to generate large collections---a.k.a. ``ensembles"---of legally valid plans, in order to create baseline statistical profiles to which any particular plan may be compared with regard to measures of interest.  (See, e.g., \cite{ACHHM21}, \cite{DDS19}, \cite{DW22}.)  Plans in an ensemble are comprised of districts that are built from smaller units; in the case of districting plans for U.S. states, these units are usually either U.S. census blocks or voting precincts.

Because most jurisdictions require districts to be contiguous, one of the first and most fundamental tasks that an algorithm must perform is to extract {\em adjacency relations} between geographical units; specifically, the algorithm must be able to accurately discern whether any pair of units share a boundary of positive length.
Unfortunately, in practice the available map data for unit boundaries contains frequent errors in the forms of gaps and overlaps between units that interfere with this task. Maps of voting precincts are especially prone to these issues, as they are often patched together from county-level maps of widely varying quality to form a statewide map.  While these errors are usually small with regard to area, they often create substantial inaccuracies in the adjacency relations between units. 

Some of the most commonly used geographical software packages (e.g., ArcGIS, QGIS) provide tools for assessing and repairing gaps and overlaps between geographic units.  
The repair algorithms in these tools are fairly simple; one common variation, which we will refer to as the {\tt quick\_repair} algorithm, is that any polygon created by either an overlap between two units or a gap between units is assigned to the unit with which its boundary shares the largest perimeter.  But it turns out that this repair strategy can dramatically worsen inaccuracies in the adjacency relations between units. 

The author first became aware of this issue in 2021 while working with colleagues\footnote{Daryl DeFord of Washington State University, and Beth Malmskog and Flavia Sancier-Barbosa, both of Colorado College} as consultants for the Colorado Independent Legislative Redistricting Commission.  A map of Colorado's 2020 voting precincts was provided by the Commission staff, and after using the {\tt quick\_repair} algorithm described above to repair gaps and overlaps between precincts, we proceeded to draw random district plans based on this precinct map.\footnote{Plans were drawn using the ReCom algorithm as implemented in the MGGG Redistricting Lab's {\tt GerryChain} Python package, available at \url{https://github.com/mggg/gerrychain}.}  A spot check early in the process revealed that the algorithm was drawing plans with discontiguous districts, which in theory should not have been possible.  The specific district plan where we first noticed this phenomenon is shown in Figure \ref{discontiguous-plans-fig}.

\begin{figure}[h!]
\begin{center}
\subfloat[]{\includegraphics[height=2in]{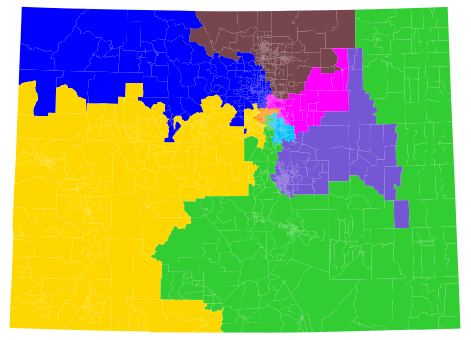} }
\ \ \ 
\subfloat[]{ \includegraphics[height=2in]{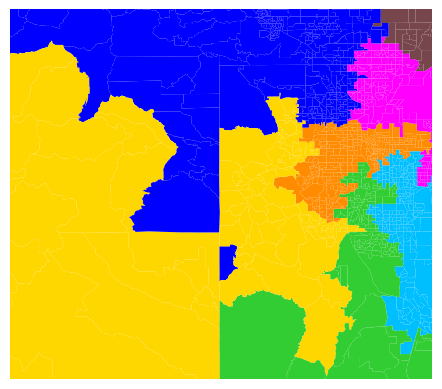} }
\end{center}
\caption{(a) Randomly drawn Colorado district plan for U.S. House based on 2020 precinct map; (b) Close-up view of a discontiguous district}
\label{discontiguous-plans-fig}
\end{figure}

Upon further investigation, we realized that the discontiguous district in the plan in Figure \ref{discontiguous-plans-fig} was made possible by the way that the {\tt quick\_repair} algorithm addressed gaps in the original precinct map. The original map contains a long, thin vertical gap along a county boundary that is adjacent to 15 precincts, shown in Figure 
\ref{gap_1_maps-fig}(a).  The {\tt quick\_repair} algorithm assigned the entire gap to the northeastern-most precinct, as shown in Figure \ref{gap_1_maps-fig}(b).  As a result, this precinct was now considered adjacent to {\em all} of the other 14 precincts adjacent to the gap.  Comparing with the discontiguous district in Figure \ref{discontiguous-plans-fig}(b), we see that the small disconnected component is precisely the southeastern-most precinct adjacent to this gap---and according to the adjacency relations obtained from the repaired precinct map, this district would be considered contiguous.
\begin{figure}[h!]
\begin{center}
\subfloat[]{\includegraphics[height=3in]{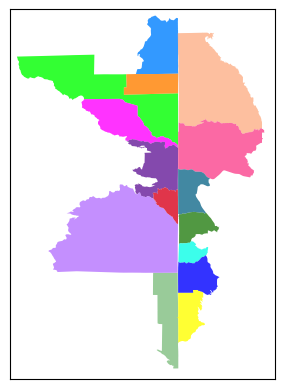} }
\hspace{0.5in}
\subfloat[]{\includegraphics[height=3in]{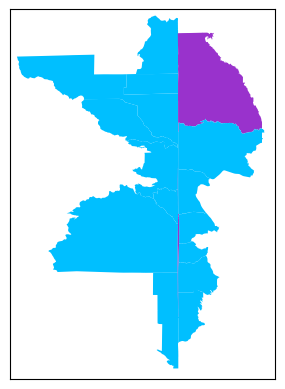} }
\end{center}
\caption{(a) 15 precincts adjacent to a single gap between counties; (b) {\tt quick\_repair} algorithm assigns the entire gap to a single precinct}
\label{gap_1_maps-fig}
\end{figure}

This particular problem was caused by extraneous adjacency relations introduced 
when the {\tt quick\_repair} algorithm
assigned the entire gap to a single precinct.  On the other hand, if we had used the original precinct map without any repairs, the district-building algorithm would have omitted all adjacencies between geographically adjacent precincts on opposite sides of the gap. Neither choice adequately represents the true adjacency relations between the precincts in this region.

\subsection{Problem statement}

In this paper we present a new algorithm, called {\tt smart\_repair},\footnote{A Python implementation of {\tt smart\_repair}, based primarily on Python's GeoPandas and Shapely libraries, is available in the MGGG Redistricting Lab's {\tt Maup} package, available at \url{https://github.com/mggg/maup}.  (This package also contains an implementation of the {\tt quick\_repair} algorithm described in Section \ref{repair-gone-wrong-subsec}.) \label{maup-footnote}} to address the following general problem:

\begin{problem}\label{big-problem}
Let $T^* = \{P^*_1, \ldots, P^*_N\}$ be a polygonal tiling of a simply connected region $R^*$ in the plane, and let $T = \{P_1, \ldots, P_N\}$ be a noisy version of $T^*$ obtained by making small perturbations to the coordinates of the vertices of the polygons in $T^*$.  In general, $T$ will only be an approximate tiling of a region $R$ that closely approximates $R^*$, due to the presence of small gaps and overlaps between the perturbed polygons in $T$. 

Given the approximate tiling $T$, construct a true tiling $\widetilde{T} = \{\widetilde{P}_1, \ldots, \widetilde{P}_N\}$ of $R$ that closely approximates $T^*$, in the sense that:
\begin{enumerate}
\item For each $k = 1, \ldots, N$, the area of $\widetilde{P}_k \cap P^*_k$ is as large as possible.
\item For each $j,k = 1, \ldots, N$, the intersection of the boundaries $\partial \widetilde{P}_j \cap \partial \widetilde{P}_k$ is a path of positive length if and only if the intersection of the boundaries $\partial P^*_j \cap \partial P^*_k$ is a path of positive length; i.e., $\widetilde{T}$ accurately reproduces the {\em adjacency relations} between the polygons in $T^*$.
\end{enumerate}
\end{problem}

Since the tiling $T^*$ is unknown, the objectives above cannot generally be achieved with complete certainty; the goal is to infer as much information as possible about the polygons in $T^*$ and their adjacency relations from the polygons in $T$.  

\subsection{Outline of the  {\tt smart\_repair} algorithm}

In order to address Problem \ref{big-problem}, the {\tt smart\_repair} algorithm  constructs the tiling $\widetilde{T}$ via the following steps, which will be described more fully in Section \ref{main-alg-sec}:
\begin{enumerate}
\item {\bf Construct refined tiling:} The union of the polygon boundaries $\{\partial P_1, \ldots, \partial P_N\}$ forms a simplicial 1-complex that partitions $R$ into a set of polygons $\mathcal{P} = \{Q_1, \ldots, Q_M\}$ that tile $R$.  The tiling $\mathcal{P}$ may be thought of as a {\em refinement} of $T$, in the sense that every polygon in $T$ is a union of polygons in $\mathcal{P}$. 

For each $Q \in \mathcal{P}$, there is a maximal subset $S_Q \subset \{1,\ldots, N\}$ such that
\[ Q \subset \bigcap_{k \in S_Q} P_k. \]
The cardinality $|S_Q|$ is called the {\em overlap order} of $Q$.  Polygons in $\mathcal{P}$ with overlap order 1 are each contained in exactly one polygon in $T$, while polygons with overlap order 0 represent gaps between polygons in $T$.

\item {\bf Assign overlaps:} The construction of the polygons $\widetilde{P}_1, \ldots, \widetilde{P}_N$ begins by assigning each polygon $Q \in \mathcal{P}$ of overlap order 1 to the unique polygon $\widetilde{P}_k$ for which $Q \subset P_k$. Higher-order overlaps are then assigned to polygons in $\widetilde{T}$ in increasing order, with first priority given to assignments needed to preserve polygon connectivity, and otherwise assigned to the polygon with which the overlap shares the largest perimeter.

\item {\bf Close gaps:} This is the most innovative and the most complicated step in the {\tt smart\_repair} algorithm.  Instead of each gap being assigned to a single polygon in $\widetilde{T}$, gaps are subdivided into pieces that are assigned to different polygons.  In order to reproduce the adjacency relations in the unknown tiling $T^*$ as closely as possible, this step is designed with two guiding principles in mind:
\begin{itemize}
\item Optimize the convexity of the repaired polygons.
\item For gaps with 4 or more adjacent polygons, after optimizing for convexity consider all non-adjacent pairs of polygons that are {\em strongly mutually visible} (cf. Definition \ref{smv-def}) across the gap. Among all such pairs, the polygons with the shortest distance between them should become adjacent after the gap is repaired. 
\end{itemize}

\end{enumerate}

There are also two optional features in the {\tt smart\_repair} algorithm:
\begin{enumerate}
\item {\bf Nesting into larger units:} In some applications, the polygons in $T$ are intended to nest cleanly into some larger units; e.g., in many states, voting precincts should nest cleanly into counties.  The user may optionally specify a clean tiling $T'$ of a region $R'$ that closely approximates $R$---e.g., a map of county boundaries within a state---and then performs the repair process so that the repaired polygons in $\widetilde{T}$ nest cleanly into the polygons in $T'$.
\item {\bf Small rook-to-queen adjacency conversion:} Whether as a result of inaccuracies in the original polygons or as an artifact of the repair algorithm, it may happen that some of the repaired polygons in $\widetilde{T}$ share boundaries with very short perimeter but should actually be considered {\em queen adjacent} (i.e., intersecting at only a single point) rather than {\em rook adjacent} (i.e., intersecting along a boundary of positive length).  To address this issue, there is an optional final
step in which all rook adjacencies of length below a user-specified parameter are converted to queen adjacencies. 
\end{enumerate}

While the development of the {\tt smart\_repair} algorithm 
was motivated by the author's work in redistricting, its potential applications to versions of Problem \ref{big-problem} extend far beyond this context.  
Geographic Information Systems (GIS) are used for an enormous variety of applications, and geospatial data is often rife with gaps, overlaps, and other problems due to rounding errors and other inaccuracies in polygon boundaries.  
As just one example, in \cite{HNP22} the authors describe similar issues that arose in a project involving topological data analysis of COVID-19 data by ZIP code.  As such, we hope that this algorithm and its Python implementation as part of the {\tt Maup} package (see Footnote \ref{maup-footnote}) will be of interest to the broader GIS community.



The remainder of the paper is organized as follows.
In Section \ref{theory-sec}, we review some facts about polygon geometry and prove some new results that will inform the {\tt smart\_repair} algorithm.
In Section \ref{main-alg-sec}, we present the primary {\tt smart\_repair} algorithm and describe how it compares with the {\tt quick\_repair} algorithm.  
The main result of this section is Theorem \ref{alg-adjacency-thm}, which shows that the adjacency relations among the repaired polygons $\widetilde{P}_1, \ldots, \widetilde{P}_N$ conform with the guiding principles described above.
In Section \ref{bells-and-whistles-sec}, we describe the optional features of the {\tt smart\_repair} algorithm.  
In Section \ref{polynomial-time-sec}, we perform a rough estimate of the runtime complexity for the {\tt smart\_repair} algorithm and show that it runs in polynomial time.
Finally, in Section \ref{back-to-beginning-sec} we conclude by illustrating how the {\tt smart\_repair} algorithm performs on the gap from Figure \ref{gap_1_maps-fig} that inspired its development and on a small selection of state-level voting precinct maps.
Detailed examples are included throughout.

\section{Polygon geometry}\label{theory-sec}

We begin with some preliminary material on polygon geometry.
Most of the background material in this section may be found in \cite{GHLST87} and \cite{LP84}.  

\begin{definition}\label{simple-poly-def}
\hspace{0.5in}
\newline
\begin{itemize}
\vspace{-0.3in}
\item A {\em polygonal path} $\overline{v_1 v_2 \cdots v_k}$ is a sequence of points $v_1, v_2, \ldots, v_k$ in the plane, called the {\em vertices} of the path, for which every pair of adjacent points $v_i$, $v_{i+1}$ ($1 \leq i \leq k-1$) represents the line segment joining $v_i$ to $v_{i+1}$, and no two non-consecutive segments intersect.
\item A {\em simple polygon} $P$ with $n$ vertices is a polygonal path $\overline{v_1 v_2 \cdots v_{n+1}}$ with $v_{n+1} = v_1$. 
\item The {\em interior angle} of a simple polygon $P$ at a vertex $v_k$ is the inward-facing angle $\theta$ between the line segments $\overline{v_{k-1} v_k}$ and $\overline{v_k v_{k+1}}$, taking values in the range $0 < \theta < 2\pi$.  
\item The {\em exterior angle} of a simple polygon $P$ at a vertex $v_k$ is the angle $\varphi = \pi - \theta$, where $\theta$ is the interior angle of $P$ at $v_k$.  The exterior angle takes values in the range $-\pi < \varphi < \pi$.
\item A vertex $v_k$ of a simple polygon $P$ is called {\em convex} if the interior angle $\theta$ of $P$ at $v_k$ satisfies $\theta < \pi$ (or equivalently, if the exterior angle $\varphi$ of $P$ at $v_k$ satisfies $\varphi > 0$) and {\em reflex} if $\theta \geq \pi$ (or equivalently, if $\varphi \leq 0$).\footnote{It is convenient for our purposes to allow the definition of ``reflex" to include vertices whose interior angle is equal to $\pi$ and not necessarily strictly greater than $\pi$.}

\end{itemize}

\end{definition}

Many of our constructions will require finding the shortest path within a simple polygon $P$ between two vertices $v_i$, $v_j$ of $P$.  This is a special case of the more general problem of finding the shortest path within $P$ between any two points either in the interior or on the boundary of $P$.  
We will use the algorithm of Lee and Preparata described in \cite{LP84} for all our shortest path constructions, and we will make frequent use of the following lemma from \cite{LP84}:

\begin{lemma}[\cite{LP84}]\label{shortest-path-lemma}
Let $v_i$, $v_j$ be two vertices of a simple polygon $P$, and let $\gamma$ be the shortest path within $P$ between $v_i$ and $v_j$.  Then all vertices of $\gamma$ are also vertices of $P$.
\end{lemma}

The next proposition describes a necessary and sufficient condition for a subset of a polygon boundary to be the shortest path within the polygon between its endpoints.

\begin{proposition}\label{convexification-prop}
Let $P$ be a simple polygon.  Let $v_i, v_j$ be distinct vertices of $P$, and let $B \subset \partial P$ be one of the two polygonal paths in $\partial P$ from $v_i$ to $v_j$.  Then $B$ is the shortest path from $v_i$ to $v_j$ in $P$ if and only if every interior vertex of $B$ is reflex.
\end{proposition}

\begin{proof}
First, suppose that $B$ is the shortest path from $v_i$ to $v_j$ in $P$.
Suppose for the sake of contradiction that $B$ has an interior vertex $v_k$ that is convex.  Then there exist $\epsilon > 0$ and points $p, q \in B$ on either side of $v_k$ at distance $\epsilon$ from $v_k$ such that the line segment $\overline{pq}$ is contained in $P$ and intersects $B$ only at its endpoints.  By the triangle inequality, the line segment $\overline{pq}$ is a shorter path from $p$ to $q$ in $P$ than the pair of line segments $\overline{pv_k}$, $\overline{v_k q}$.  Since the shortest path assumption on $B$ implies that the shortest path in $P$ between any pair of points in $B$ is contained in $B$, this is a contradiction; therefore every interior vertex of $B$ is reflex.

Conversely, suppose that every interior vertex of $B$ is reflex.  It is straightforward to check that the algorithm in \cite{LP84} for computing shortest paths in polygons will construct $B$ as the shortest path between its endpoints.
\end{proof}

The condition in Proposition \ref{convexification-prop} is useful enough to give it a name; following \cite{GHLST87}, we make the following definition:

\begin{definition}\label{outward-convex-def}
Let $P$ be a simple polygon.  A polygonal path $B \subset \partial P$ is called {\em outward convex} if every interior vertex of $B$ is reflex.
\end{definition}

We will also be interested in the notion of {\em visibility} between subsets of a polygon boundary.  Specifically, the following notion will play an important role in the {\tt smart\_repair} algorithm:

\begin{definition}\label{smv-def}
Let $P$ be a simple polygon, and let $B', B'' \subset \partial P$ be polygonal paths within $\partial P$.  We will say that the pair $(B', B'')$ is {\em strongly mutually visible} in $P$ if there exist points $p'$ in the interior of $B'$ and $p''$ in the interior of $B''$ such that the line segment $\overline{p' p''}$ is contained in $P$ and intersects $\partial P$ only at its endpoints.

Similarly, we will say that a point $p' \in \partial P$ is {\em strongly visible} in $P$ from a point $p'' \in \partial P$ if the line segment $\overline{p' p''}$ is contained in $P$ and intersects $\partial P$ only at its endpoints.
\end{definition}

Note that the containment condition on the line segment $\overline{p' p''}$ in Definition \ref{smv-def} is an open condition on the points $p' \in B'$, $p'' \in B''$, so strong mutual visibility of a pair of paths $(B', B'')$ in $P$ implies that there exist nonempty open subsets $U' \subset B'$ and $U'' \subset B''$ such that for every pair of points $(q', q'')$ with $q' \in U'$ and $q'' \in U''$, $q'$ and $q''$ are strongly visible to each other in $P$.

\begin{proposition}\label{smv-prop}
Let $P$ be a simple polygon, and let $m \geq 4$.  Let $\partial P$ be partitioned into polygonal paths $B_1, \ldots, B_m$ such that 
\begin{gather*}
\bigcup_{k=1}^m B_k = \partial P, \\
B_k \cap B_{k+1} = \{v_k\}, \qquad 1 \leq k \leq m-1, \\
B_m \cap B_1 = \{v_m\}
\end{gather*}
for distinct vertices $v_1, \ldots, v_m$ of $P$.  
Suppose that each of the paths $B_1, \ldots, B_m$ is outward convex. Then there exists at least one non-adjacent pair $(B_i, B_j)$ (i.e., a pair satisfying $B_i \cap B_j = \emptyset$) that is strongly mutually visible.

\end{proposition}

For example, the polygon boundaries in Figure \ref{poly-5-subbound-fig} have each been partitioned into 5 polygonal paths satisfying the conditions of Proposition \ref{smv-prop}.  For each of these polygons, all non-adjacent pairs $(B_i, B_j)$ are strongly mutually visible except for $(B_1,B_3)$.
\begin{figure}[h!]
\begin{center}
\includegraphics[height=2in]{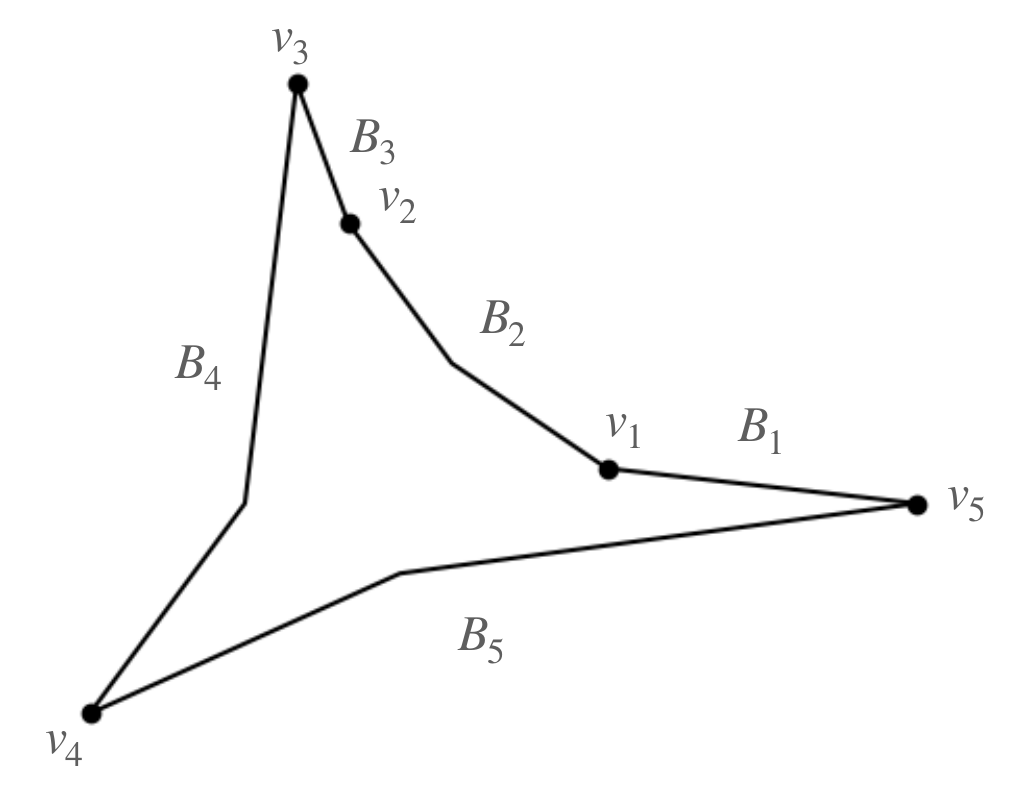} 
\ \ \ 
\includegraphics[height=2in]{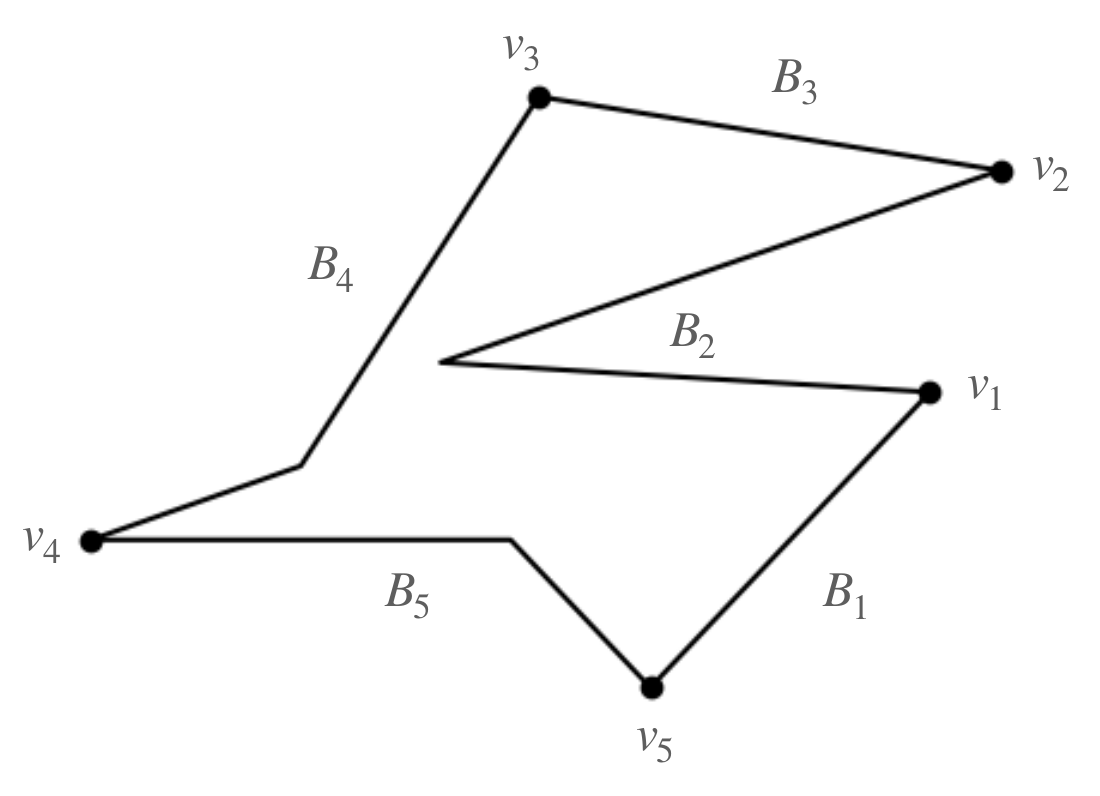} 
\end{center}
\caption{Polygon boundaries each partitioned into 5 outward convex polygonal paths; in both polygons, all non-adjacent pairs $(B_i,B_j)$ except for $(B_1, B_3)$ are strongly mutually visible.}
\label{poly-5-subbound-fig}
\end{figure}

\begin{proof}
Every simple polygon contains at least 3 convex vertices; this is an immediate consequence of the fact that every simple polygon has exterior angles that sum to $2\pi$, and each exterior angle is strictly less than $\pi$ and greater than $0$ if and only if the corresponding vertex is convex.  Moreover, since each of the paths $B_1, \ldots, B_m$ is outward convex, the convex vertices of $P$ must be a subset of the intersection points $\{v_1, \ldots, v_m\}$ between these paths.

Let $v_k$ be a convex vertex of $P$ with interior angle $\theta$. For $0 < \alpha < \theta$, let $\ell_\alpha$ be the line segment in $P$ that starts at $v_k$, forms an interior angle of $\alpha$ with the edge of $\partial P$ emanating from $v_k$, and ends at the first point $q_\alpha \in \partial P$ where it intersects $\partial P$ at a positive distance from $v_k$.
By construction, $\ell_\alpha$ is not contained in either $B_k$ or $B_{k+1}$, and $\ell_\alpha$ is the shortest path from $v_k$ to $q_\alpha$; therefore, Proposition \ref{convexification-prop} implies that $q_\alpha \notin B_k \cup B_{k+1}$. It follows that $q_\alpha \in B_{j_\alpha}$ for some $j_\alpha \neq k, k+1$.  
As $\alpha$ varies continuously between $0$ and $\theta$, there must be some nonempty, open subinterval $(a,b) \subset (0, \theta)$ on which $j_\alpha$ is equal to a constant value $j_0$, and for which the points $\{q_\alpha \mid \alpha \in (a,b)\}$ form an open subset of $B_{j_0}$.

Since $m \geq 4$, the path $B_{j_0}$ is not adjacent to at least one of $B_k$ and $B_{k+1}$.  Without loss of generality, suppose that $B_{j_0}$ is not adjacent to $B_k$, and choose $\alpha$ such that $q_\alpha$ is an interior point of $B_{j_0}$.  Since the interior of the line segment $\overline{v_k q_\alpha}$ is contained within the interior of $P$, there exists $\epsilon > 0$ such that the line segment $\ell'$ between $q_\alpha$ and the interior point of $B_k$ at distance $\epsilon$ from $v_k$ also has its interior contained within the interior of $P$.  Therefore, $B_k$ and $B_{j_0}$ are strongly mutually visible.
\end{proof}

Figure \ref{illustrate-smv-prop-proof-fig} illustrates this construction starting from the vertex $v_3$.  The family of line segments $\ell_\alpha$ (shown in green) intersects open subsets of both $B_2$ and $B_5$. Shifting the initial point of one of the segments that intersects $B_2$ a small distance along $B_4$ produces a line segment $\ell'$ (shown in magenta) in $P$ that joins an interior point of $B_4$ to an interior point of $B_2$ and only intersects $\partial P$ at its endpoints.
\begin{figure}[h!]
\begin{center}
\includegraphics[height=2in]{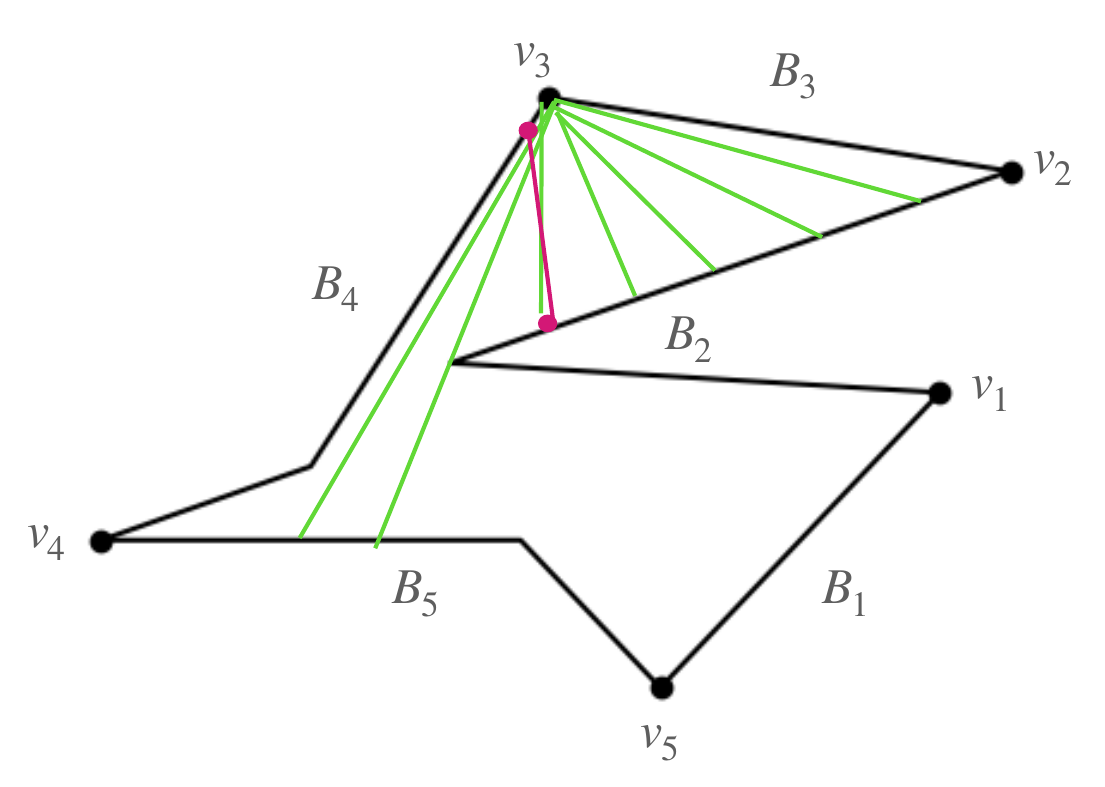} 
\end{center}
\caption{Illustration of construction in the proof of Proposition \ref{smv-prop}}
\label{illustrate-smv-prop-proof-fig}
\end{figure}

\begin{proposition}\label{tri-visibility-prop}
Let $P$ and $B_1, \ldots, B_m$ be as in Proposition \ref{smv-prop} with $m=3$.
Then for any pair $(B_i, B_j)$ with $1 \leq i,j \leq 3$, $i \neq j$, every interior point of $B_i$ is strongly visible from some interior point of $B_j$.
\end{proposition}

\begin{proof}
As noted in the proof of Proposition \ref{smv-prop}, all 3 of the vertices $v_1$, $v_2$, $v_3$ must be convex.
Without loss of generality, let $i=1$, $j=2$, and let $p \in B_1$ be in interior point.  
Let $\ell$ be a line segment in $P$ passing through $p$ and tangent to $B_1$.  Outward convexity of $B_1$ implies that such a line segment $\ell$ exists, and that the extension $\widetilde{\ell}$ of $\ell$ in either direction from $p$ must intersect either $B_2$ or $B_3$.  Furthermore, $\widetilde{\ell}$ cannot intersect either $B_2$ or $B_3$ twice, since that would contradict the fact that $B_2$ and $B_3$ are each the shortest paths between their endpoints.  Therefore $\widetilde{\ell}$ must intersect both $B_2$ and $B_3$.

Outward convexity of $B_3$ implies that $\widetilde{\ell}$ cannot intersect $B_2$ at $v_2$, so it must intersect $B_2$ at some positive distance from $v_2$.  Therefore, a slight rotation of $\widetilde{\ell}$ about $p$ produces a line segment $\ell'$ that intersects $B_1$ at $p$ and $B_2$ at some interior point $q$, and whose interior is contained in the interior of $P$. Then by definition, $p$ is strongly visible from $q$, as desired.
This construction is illustrated in Figure \ref{illustrate-tri-visibility-prop-proof-fig}; the line segment $\widetilde{\ell}$ is shown in green, and the line segment $\ell'$ connecting $p$ and $q$ is shown in magenta.
\end{proof}

\begin{figure}[h!]
\begin{center}
\includegraphics[height=2in]{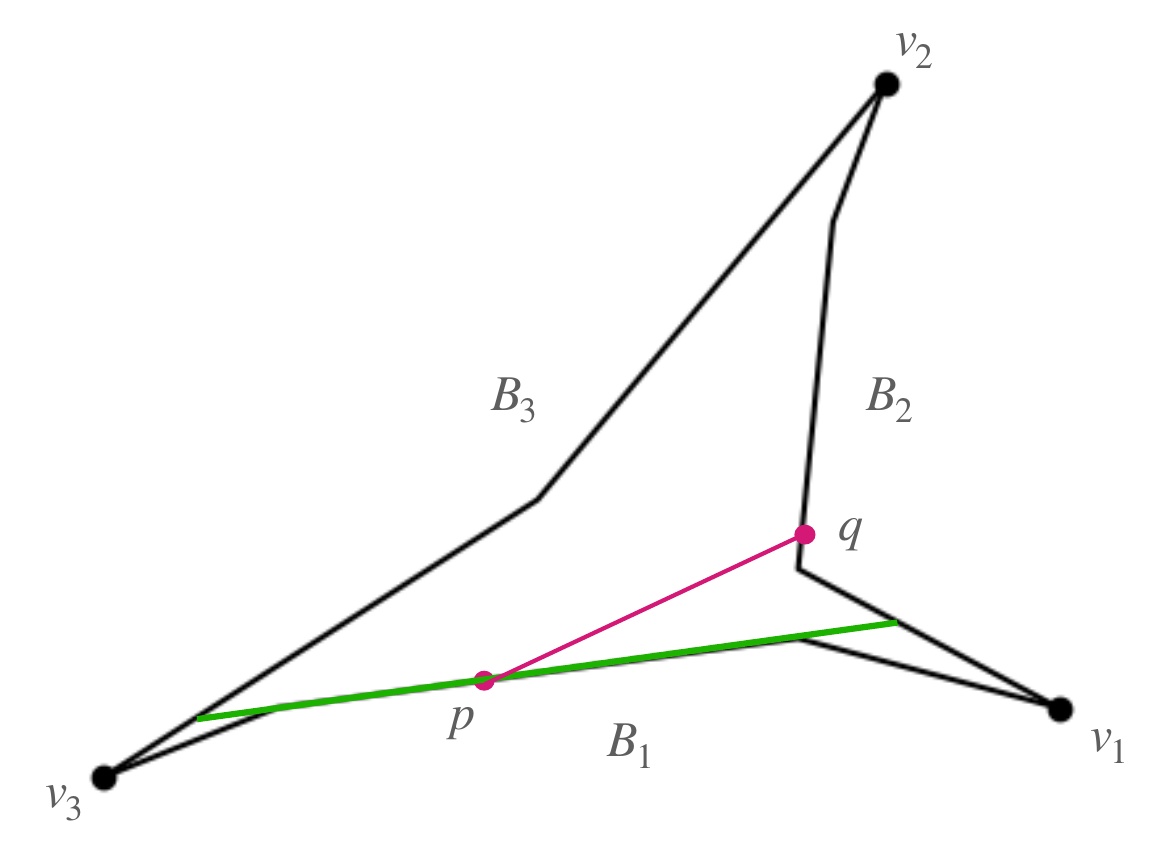} 
\end{center}
\caption{Illustration of construction in the proof of Proposition \ref{tri-visibility-prop}}
\label{illustrate-tri-visibility-prop-proof-fig}
\end{figure}

\begin{proposition}\label{quad-visibility-prop}
Let $P$ and $B_1, \ldots, B_m$ be as in Proposition \ref{smv-prop} with $m=4$.  Then both non-adjacent pairs $(B_1, B_3)$ and $(B_2, B_4)$ are strongly mutually visible.
\end{proposition}

\begin{proof}
As noted in the proof of Proposition \ref{smv-prop}, at least 3 of the 4 vertices $v_1$, $v_2$, $v_3$, $v_4$ must be convex.  Without loss of generality, assume that $v_1$, $v_2$, and $v_3$ are convex, and let $\beta$ be the shortest path in $P$ from $v_2$ to $v_4$.  Since $v_1$ and $v_3$ are convex, Proposition \ref{convexification-prop} implies that neither $B_1 \cup B_2$ nor $B_3 \cup B_4$ can be the shortest path from $v_2$ to $v_4$; therefore $\beta$ contains a line segment $\ell$ whose endpoints are contained in $\partial P$ and whose interior is contained in the interior of $P$.  Outward convexity of $B_1, \ldots, B_4$ then implies that the endpoints of $\ell$ are not both contained in the same path $B_k$.

Up to symmetry, there are two possibilities:
\begin{enumerate}
\item The endpoints of $\ell$ are contained in $B_1$ and $B_3$ (or equivalently, in $B_2$ and $B_4$).  If the endpoints of $\ell$ are interior points of $B_1$ and $B_3$, then the pair $(B_1, B_3)$ is strongly mutually visible by definition.  If $v_2$ and/or $v_4$ is an endpoint of $\ell$, then a slight rotation of $\ell$ about its midpoint produces a line segment that can be extended to intersect both $B_1$ and $B_3$ at interior points; hence the pair $(B_1, B_3)$ is strongly mutually visible.  

Next, note that since $\ell$ is part of the shortest path $\beta$, it must be tangent to both $B_1$ and $B_3$.  Therefore a slight rotation of $\ell$ about its midpoint in the other direction produces a line segment whose extension does not intersect either $B_1$ or $B_3$, and which must therefore intersect both $B_2$ and $B_4$ at interior points.  Thus the pair $(B_2, B_4)$ is strongly mutually visible.
\item The endpoints of $\ell$ are interior points of $B_1$ and $B_2$ (or equivalently, of $B_3$ and $B_4$).  Then $\beta$ contains open subsets of both $B_1$ and $B_2$, and $\ell$ divides $P$ into two simple polygons, one of which is bounded by $B_3$, $B_4$, and $\beta$.  Let $P'$ denote this polygon.  Since $\beta$ is the shortest path between its endpoints in this polygon, Proposition \ref{convexification-prop} implies that $\beta$ is outward convex, and therefore $P'$ satisfies the hypotheses of Proposition \ref{smv-prop} with $m=3$.  
Applying Proposition \ref{tri-visibility-prop} to $P'$ implies that $B_3$ and the open subset of $B_1$ contained in $\beta$ are strongly mutually visible, as are $B_4$ and the open subset of $B_2$ contained in $\beta$.  Therefore, both non-adjacent pairs $(B_1, B_3)$ and $(B_2, B_4)$ are strongly mutually visible in $P'$ and hence in $P$.
\end{enumerate}
Note that it is not possible for the endpoints of $\ell$ to be contained in $B_1$ and $B_4$ (or equivalently, in $B_2$ and $B_3$), because both $B_1$ and $B_4$ are the shortest paths between each of their endpoints and the vertex $v_4$, so this would violate the shortest path property of $\beta$.
\end{proof}

Figure \ref{illustrate-quad-visibility-prop-proof-fig} illustrates the constructions in both cases above; the shortest paths $\beta$ are shown in green, and line segments between interior points of non-adjacent polygonal paths are shown in magenta.

\begin{figure}[h!]
\begin{center}
\subfloat[]{\includegraphics[height=2in]{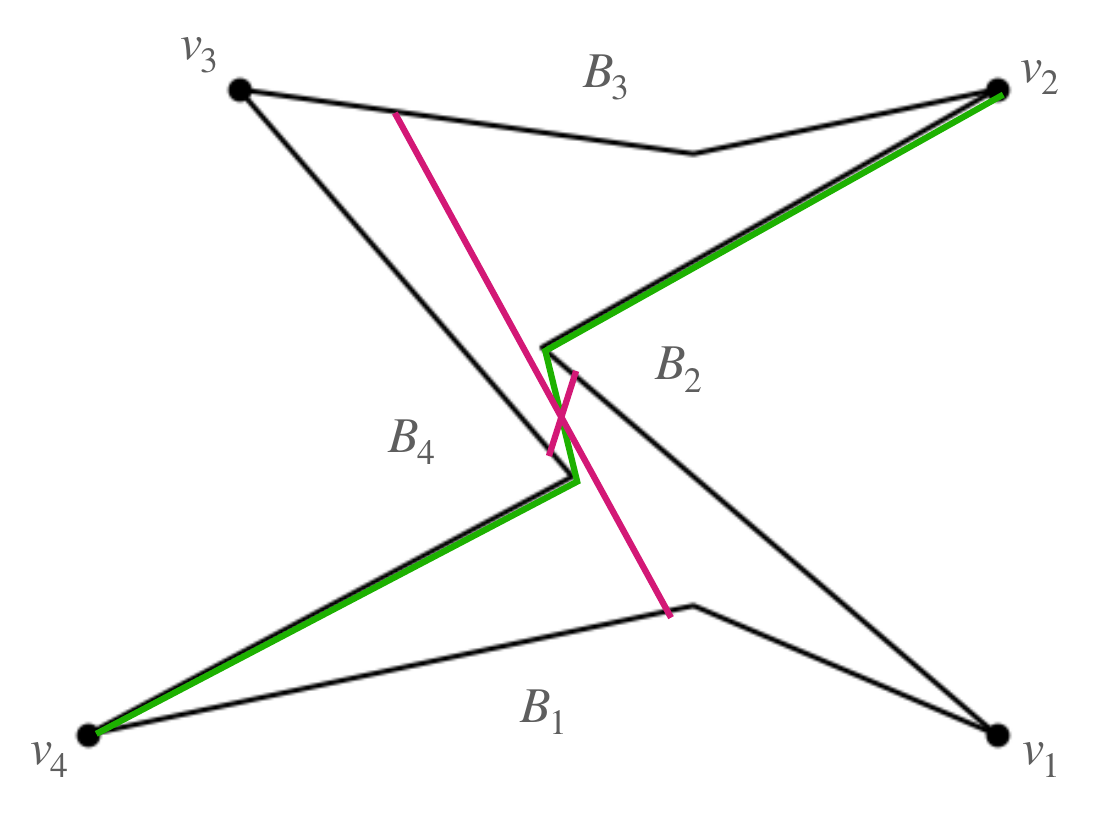}  }
\ \ \ 
\subfloat[]{\includegraphics[height=2in]{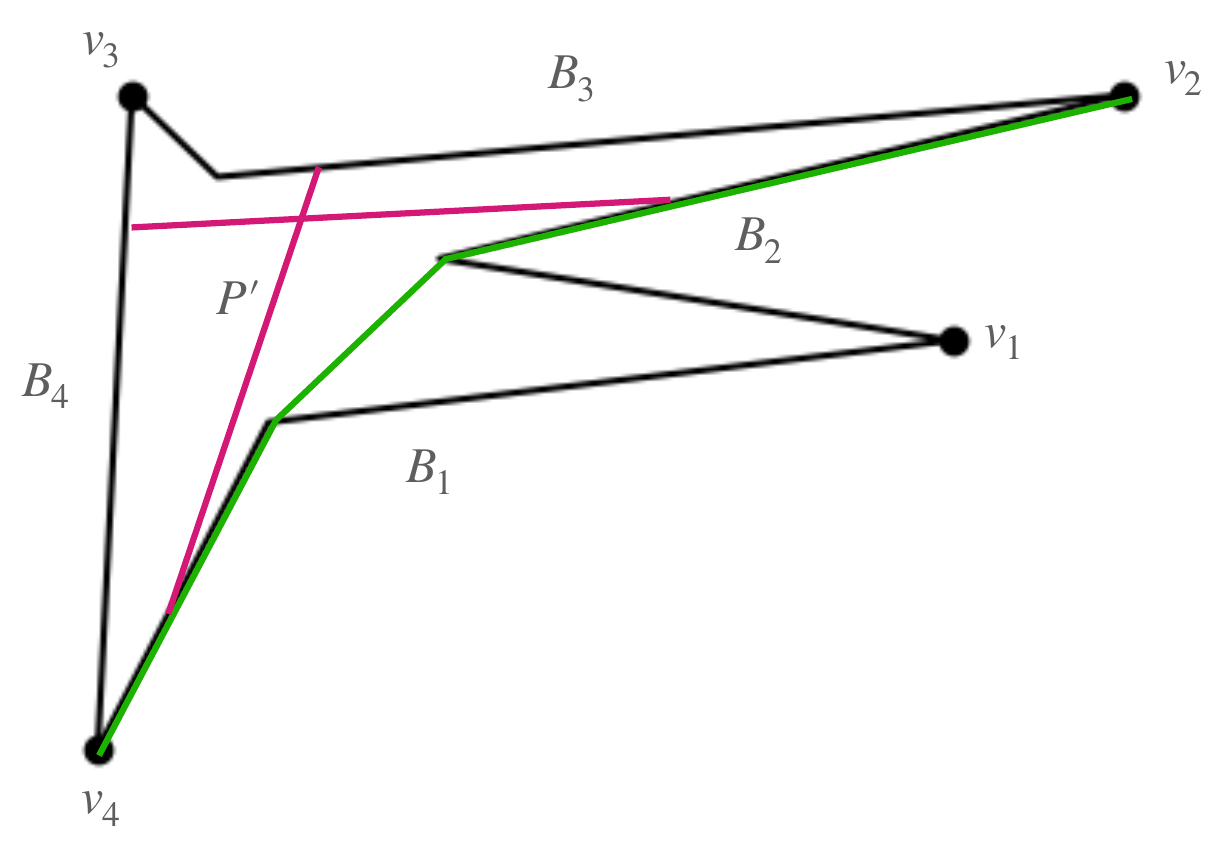}  }
\end{center}
\caption{Illustrations of constructions in the proof of Proposition \ref{quad-visibility-prop}: (a) Case 1; (b) Case (2)}
\label{illustrate-quad-visibility-prop-proof-fig}
\end{figure}

For $m \geq 5$, determining directly whether a particular non-adjacent pair $(B_i, B_j)$ is strongly mutually visible can be computationally challenging; the following theorem provides a relatively simple criterion that will be useful for the implementation of the {\tt smart\_repair} algorithm.

\begin{theorem}\label{smv-criterion-thm}
Let $P$ and $B_1, \ldots, B_m$ be as in Proposition \ref{smv-prop}, and let $(B_i, B_j)$ be a non-adjacent pair.  Orient the paths $B_i$ and $B_j$ consistently with their orientations as subsets of $\partial P$, with the standard counterclockwise orientation on $\partial P$.  Let $\alpha_1$ be the shortest path in $P$ from the terminal point of $B_i$ to the initial point of $B_j$, and let $\alpha_2$ be the shortest path in $P$ from the terminal point of $B_j$ to the initial point of $B_i$.  Then the pair $(B_i, B_j)$ is strongly mutually visible if and only if $\alpha_1$ and $\alpha_2$ are disjoint.
\end{theorem}

\begin{proof}
First suppose that the pair $(B_i, B_j)$ is strongly mutually visible. Let $p_i \in B_i$, $p_j \in B_j$ be interior points of $B_i$ and $B_j$ such that the line segment $\overline{p_i p_j}$ is contained in $P$ and intersects $\partial P$ only at its endpoints.  Then the line segment $\overline{p_i p_j}$ divides $P$ into two simple polygons $P_1, P_2$, with the endpoints of $\alpha_1$ contained in $\partial P_1$ and the endpoints of $\alpha_2$ contained in $\partial P_2$.  The triangle inequality implies that $\alpha_1$ and $\alpha_2$ cannot cross the line segment $\overline{p_i p_j}$, and so must be contained in $P_1$ and $P_2$, respectively.  Furthermore, Lemma \ref{shortest-path-lemma} implies that no vertices of either $\alpha_1$ or $\alpha_2$ are contained in $\overline{p_i p_j}$, and so 
\[ \alpha_1 \cap \overline{p_i p_j} = \alpha_2 \cap \overline{p_i p_j} = \emptyset. \]
Since $\alpha_1 \subset P_1$, $\alpha_2 \subset P_2$, and $P_1 \cap P_2 = \overline{p_1 p_2}$, it follows that $\alpha_1 \cap \alpha_2 = \emptyset$.

Conversely, suppose that $\alpha_1 \cap \alpha_2 = \emptyset$.  Then the paths $B_i, \alpha_1, B_j, \alpha_2$ together form the boundary of a simple polygon $P' \subset P$.  Since $\alpha_1$ and $\alpha_2$ are the shortest paths between their endpoints in $P$, they are also the shortest paths between their endpoints in $P'$.  By Proposition \ref{convexification-prop}, $\alpha_1$ and $\alpha_2$ are both outward convex in $P'$.  It then follows from Proposition \ref{quad-visibility-prop} that the pair $(B_i, B_j)$ is strongly mutually visible in $P'$, and hence in $P$ as well.
\end{proof}

As an illustration, consider the polygons shown in Figure \ref{poly-5-subbound-fig}.  For the pairs $(B_1, B_3)$ that are not strongly mutually visible, it is easily seen that the corresponding paths $\alpha_1, \alpha_2$ are not disjoint. In the first polygon, $\alpha_1 = B_2$ and $\alpha_2 = B_1 \cup B_2 \cup B_3$, so their intersection is the polygonal path $B_2$.  In the second polygon, $\alpha_1 = B_2$ and $\alpha_2$ intersects $B_2$ at its interior vertex, so $\alpha_1$ and $\alpha_2$ are not disjoint.  Conversely, the paths $\alpha_1, \alpha_2$ corresponding to any other non-adjacent pair in either of these polygons are easily seen to be disjoint.

Next we consider the ``diagonals" of a strongly mutually visible pair $(B_i, B_j)$.

\begin{theorem}\label{intersecting-diagonals-thm}
Let $P$ and $B_1, \ldots, B_m$ be as in Proposition \ref{smv-prop}, and let $(B_i, B_j)$ be a non-adjacent pair that are strongly mutually visible.  With $B_i$ and $B_j$ oriented as in Theorem \ref{smv-criterion-thm}, 
let $\beta_1$ be the shortest path in $P$ between the initial points of $B_i$ and $B_j$, and let $\beta_2$ be the shortest path in $P$ between the terminal points of $B_i$ and $B_j$.
Then $\beta_1$ and $\beta_2$ intersect at a single point, which is either an interior point of $P$ or a vertex of $P$.  
\end{theorem}

\begin{proof}
let $\alpha_1, \alpha_2$ be as in Theorem \ref{smv-criterion-thm}, and let $P'$ be the polygon bounded by $B_i$, $\alpha_1$, $B_j$, and $\alpha_2$.
It is straightforward to show that $\beta_1$ and $\beta_2$ are contained in $P'$, so without loss of generality we may assume that $P = P'$ and $m=4$.  Then for simplicity, set $B_i = B_1$ and $B_j = B_3$.

It is a basic topological fact that $\beta_1 \cap \beta_2$ cannot be empty. Moreover, $\beta_1 \cap \beta_2$ cannot contain more than one connected component, as this would violate the uniqueness of the shortest path between the endpoints of these components.  So $\beta_1 \cap \beta_2$ must be either a single point or a continuous sub-path of $\beta_1$ and $\beta_2$.

Consider $\beta_1$ and $\beta_2$ as paths starting at the endpoints of $B_1$.
Let $p$ be the first point of intersection between $\beta_1$ and $\beta_2$.  If $p$ is an interior point of $P$, then since every vertex of $\beta_1$ and $\beta_2$ is a vertex of $P$, $p$ must be an interior point of an edge on both $\beta_1$ and $\beta_2$.  It follows that these edges must intersect transversely at $p$, and hence $\beta_1 \cap \beta_2 = \{p\}$.  

Now suppose that $p \in \partial P$, and suppose for the sake of contradiction that $p$ is an interior point of an edge in $P$. Then $\beta_1$ and $\beta_2$ cannot intersect transversely at $p$, and hence $p$ must be a vertex of at least one of $\beta_1$ and $\beta_2$. But Lemma \ref{shortest-path-lemma} implies that every vertex of $\beta_1$ and $\beta_2$ is a vertex of $P$; therefore $p$ must be a vertex of $P$.  

Next, suppose for the sake of contradiction that $\beta_1 \cap \beta_2$ is a continuous path of positive length, and let $p, q$ be the endpoints of $\beta_1 \cap \beta_2$.  The argument in the previous paragraph shows that $p, q$ are vertices of $P$.  Additionally, $p$ and $q$ must both be contained in either $B_2$ or $B_4$, as any other possibility would contradict the shortest path property of either $\beta_1$ or $\beta_2$.  Therefore, since $\beta_1 \cap \beta_2$ is the shortest path in $P$ between $p$ and $q$, the entire path $\beta_1 \cap \beta_2$ must be contained in either $B_2$ or $B_4$.  Without loss of generality, suppose that $\beta_1 \cap \beta_2 \subset B_2$.  Uniqueness of shortest paths then implies that $\beta_1$ is coincident with $B_2$ between $p$ and $v_2$, and $\beta_2$ is coincident with $B_2$ between $v_1$ and $q$.  (See Figure \ref{impossible-diagonals-fig}.)

Now consider the polygon $P'' \subset P$ bounded by $B_4$, $\beta_1$, and $\beta_2$.  Since $\beta_1 \subset \partial P''$ is the shortest path between its endpoints and $p$ is an interior vertex of $\beta_1$, it must be reflex in $P''$.  Likewise, $q$ is an interior vertex of $\beta_2$ and so must be reflex in $P''$.  But $P''$ must have at least 3 convex vertices, and its only possible convex vertices are $v_3$, $v_4$, $p$, and $q$.  So at least one of $p, q$ must be a convex vertex of $P''$, but this is a contradiction.  Therefore, $\beta_1 \cap \beta_2$ cannot be a path of positive length and so must consist of a single point.
\end{proof}
\begin{figure}[h!]
\begin{center}
\includegraphics[height=2in]{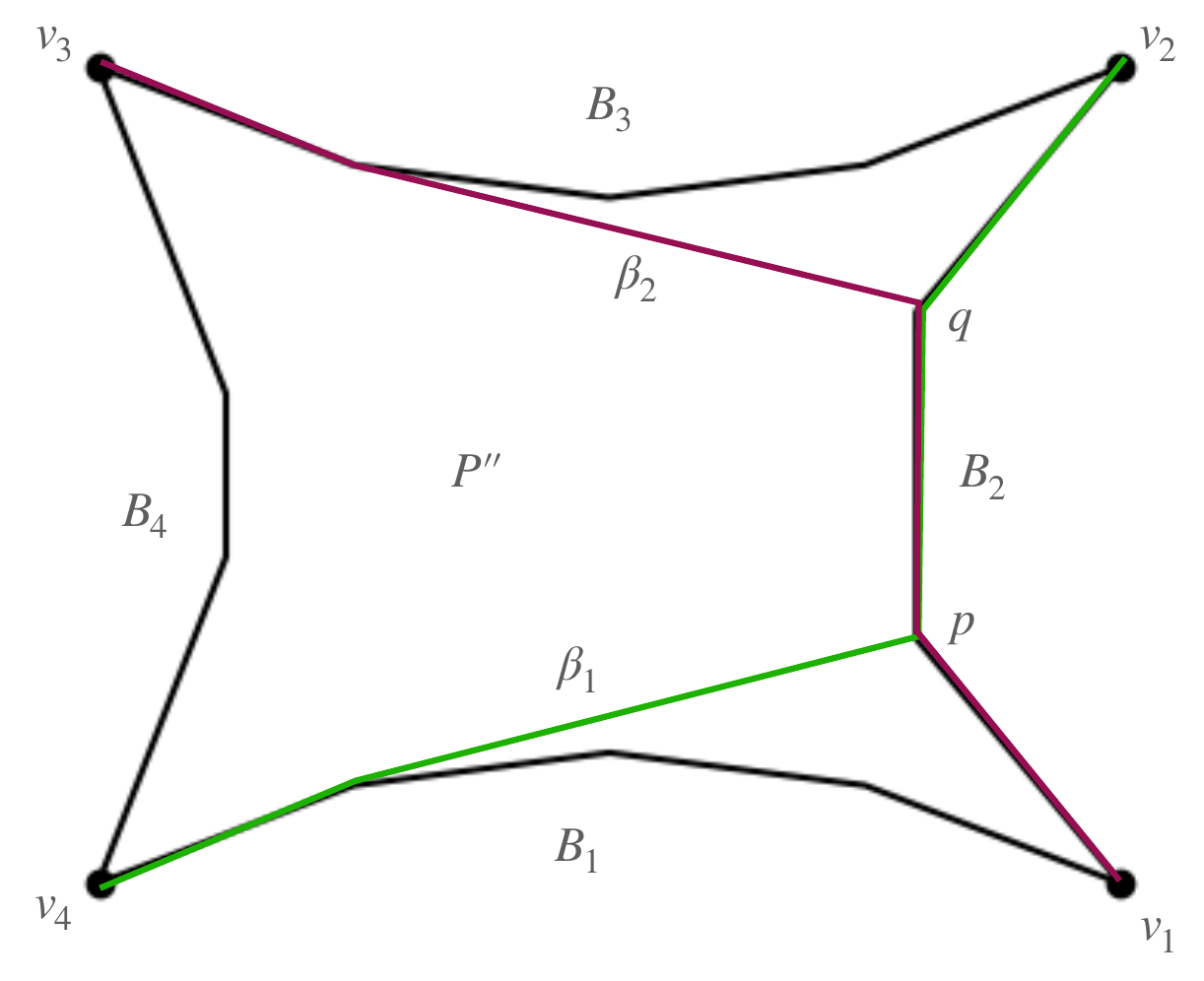} 
\end{center}
\caption{An impossible configuration of shortest paths}
\label{impossible-diagonals-fig}
\end{figure}

Finally, we prove the following theorem that will play an important role in the 
{\tt smart\_repair} algorithm:

\begin{theorem}\label{positive-area-thm}
Let $P$ and $B_1, \ldots, B_m$ be as in Proposition \ref{smv-prop}, and let $(B_i, B_j)$ be a non-adjacent pair that are strongly mutually visible. 
Let $\beta_1$ and $\beta_2$ be as in Theorem \ref{intersecting-diagonals-thm}. Then the union of $B_i$, $B_j$, $\beta_1$, and $\beta_2$ bounds a region consisting of either:
\begin{enumerate}
\item two simple polygons, each of whose boundaries intersects one of the two paths $B_i$, $B_j$ in a path of positive length and is disjoint from the other, and whose intersection is a single point located at either an interior point of $P$ or a vertex of $\partial P \setminus (B_i \cap B_j)$, or
\item one simple polygon whose boundary intersects one of the two paths $B_i$, $B_j$ in a path of positive length and contains exactly one vertex of the other.
\end{enumerate}
\end{theorem}

\begin{proof}
Let $\alpha_1$, $\alpha_2$ be as in Theorem \ref{smv-criterion-thm}, and let $P'$ be the polygon bounded by $B_i$, $\alpha_1$, $B_j$, and $\alpha_2$.  As in the proof of Theorem \ref{intersecting-diagonals-thm}, without loss of generality we may assume that $P = P'$ and that $m=4$.  Then $(B_i, B_j)$ is one of the two non-adjacent pairs $(B_1, B_3)$ or $(B_2, B_4)$.  Without loss of generality, set $B_i = B_1$ and $B_j = B_3$.

Let $p \in P$ be the intersection point of $\beta_1$ and $\beta_2$.  
\begin{itemize}
\item If $p$ is an interior point of $P$ or an interior vertex of $B_2$ or $B_4$, then $B_1$, $B_3$, $\beta_1$, and $\beta_2$ bound 2 polygons as in the first option above.  (See Figure \ref{diagonal-triangles-fig}(a)-(b).)
\item If $p$ is a vertex of $B_1$ or $B_3$, then $B_1$, $B_3$, $\beta_1$, and $\beta_2$ bound 1 polygon as in the second option above.  (See Figure \ref{diagonal-triangles-fig}(c).)
\end{itemize}
\end{proof}

\begin{figure}[h!]
\begin{center}
\subfloat[]{\includegraphics[height=1.5in]{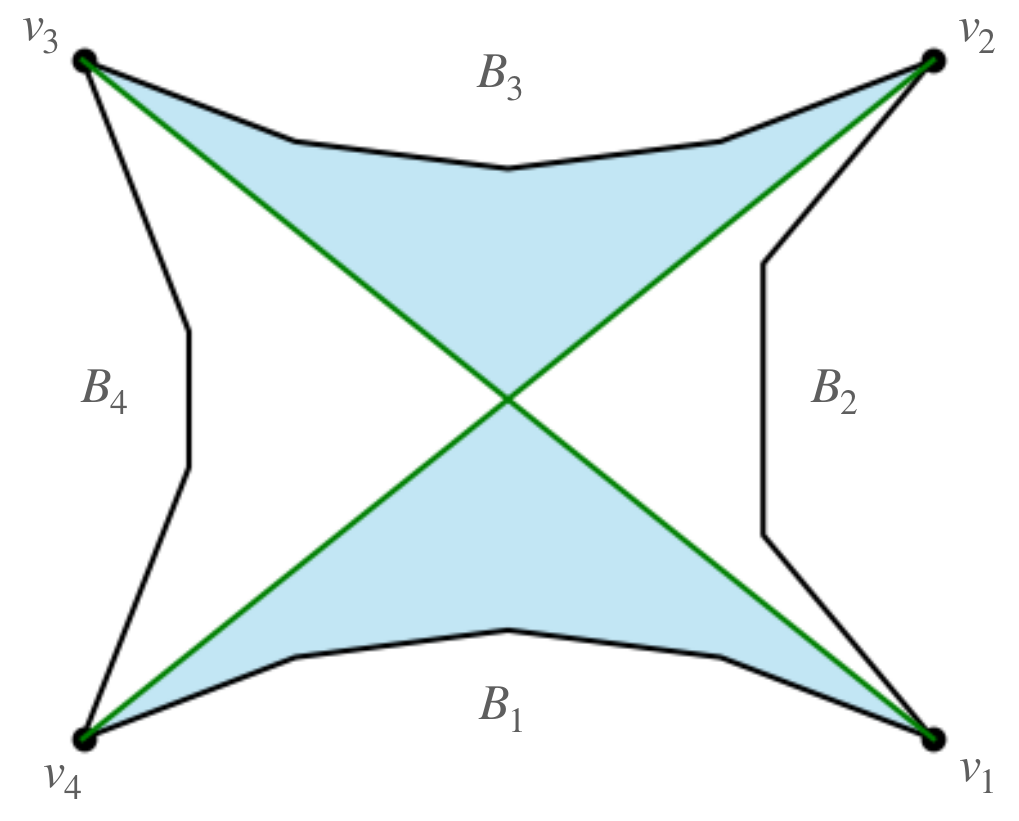}  }
\ \ \ 
\subfloat[]{\includegraphics[height=1.5in]{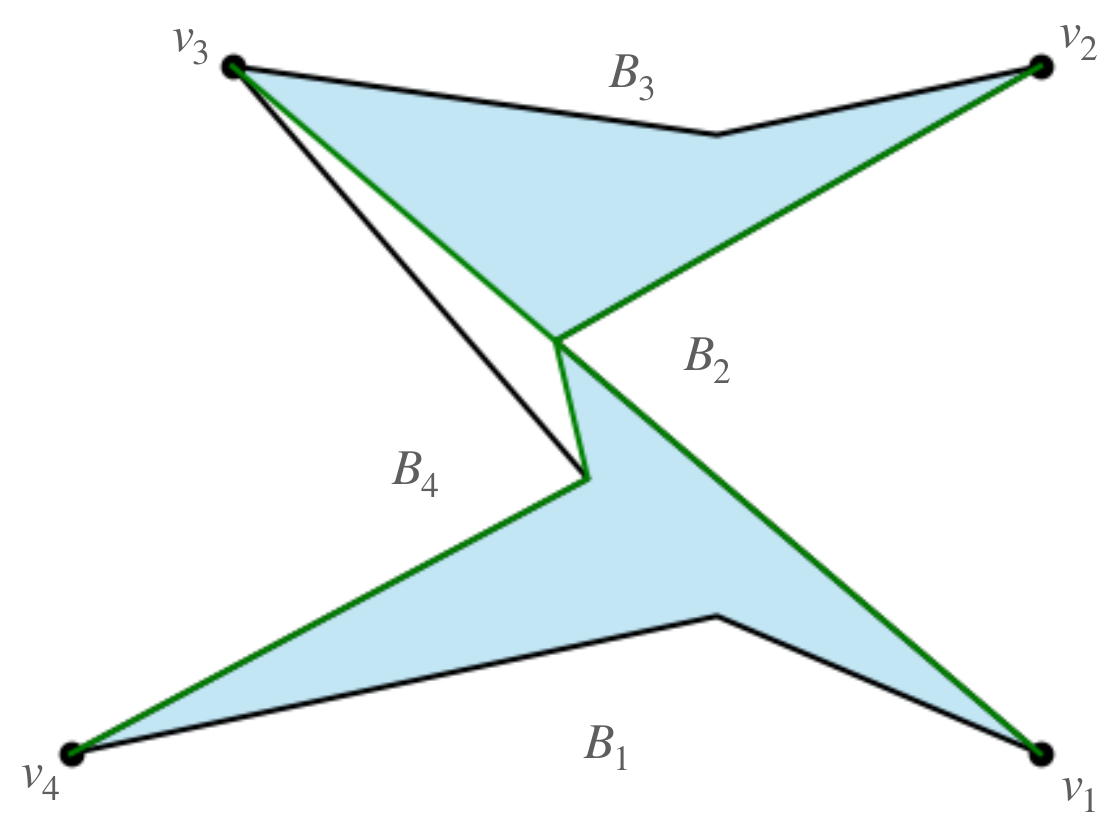}  }
\ \ \ 
\subfloat[]{\includegraphics[height=1.5in]{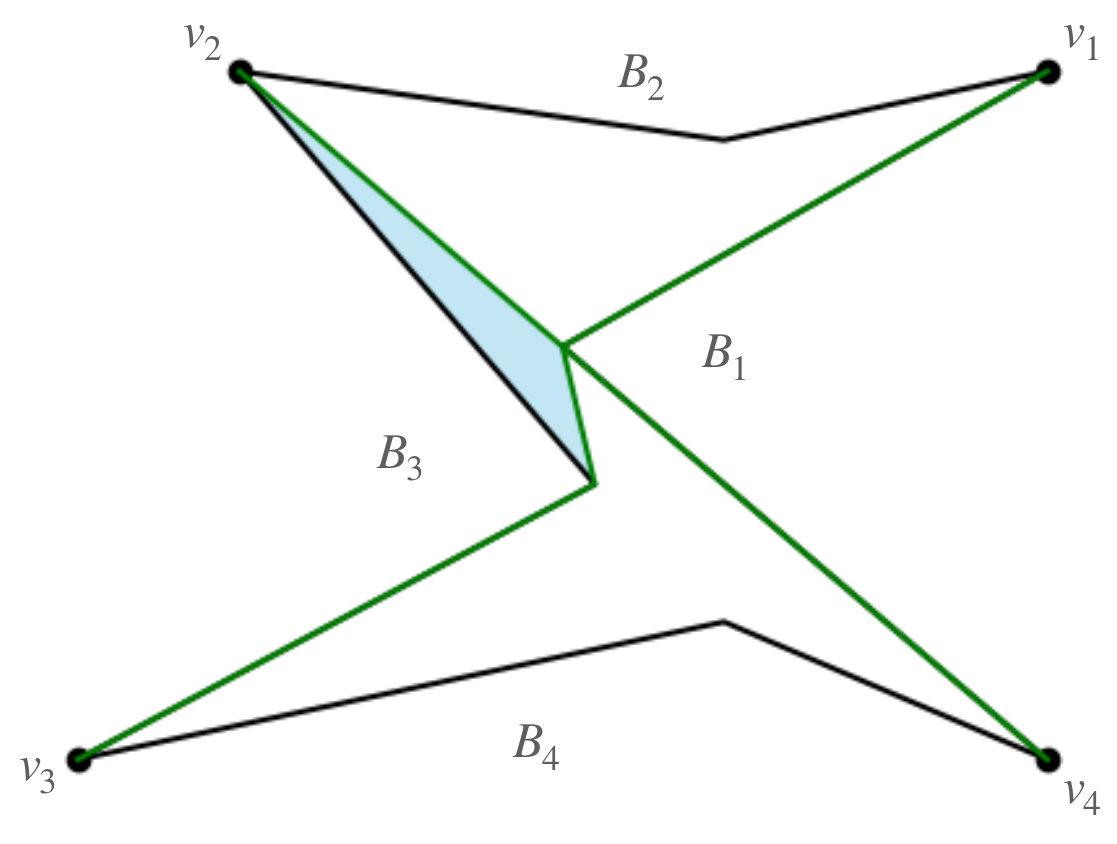}  }
\end{center}
\caption{}
\label{diagonal-triangles-fig}
\end{figure}

\section{Primary {\tt smart\_repair} algorithm}\label{main-alg-sec}

In this section, we present the details of the primary {\tt smart\_repair} algorithm and compare it with the {\tt quick\_repair} algorithm; optional features will be described in section \ref{bells-and-whistles-sec}.

Let $T = \{P_1, \ldots, P_N\}$ be a set of polygons as described in Problem \ref{big-problem}.  We will refer to the polygons $P_1,\ldots, P_N$, their repaired versions $\widetilde{P}_1, \ldots, \widetilde{P}_N$, and the intermediate stages in the construction of the repaired versions as {\em units} of $T$ and $\widetilde{T}$, respectively; this is primarily to avoid confusion with smaller polygons representing gaps, overlaps, and subdivisions thereof that appear throughout the {\tt smart\_repair} algorithm.

\subsection{Step 1: Construct refined tiling}
For the {\tt quick\_repair} algorithm, it is typical for small gaps and overlaps between units to remain even after the repair.  These result primarily from rounding errors that occur when computing points of intersection between unit boundaries.  For instance, suppose that two units overlap as shown in the first plot in Figure \ref{intersection-example-fig}.  The coordinates of the points of intersection between their boundaries generally cannot be computed numerically with perfect precision, and so the boundaries of the polygon representing the overlap may not line up cleanly with the boundaries of the original units, as shown in the second plot in Figure \ref{intersection-example-fig}.  In this example, the {\tt quick\_repair} algorithm would remove the smaller polygon representing the overlap from both of the original units and reassign it only to the unit on the left.  Because of the slight inaccuracy in the computation of the new vertices on this polygon, this process leaves a small overlap between the ``repaired" units.  Similar discrepancies in the computation of gaps between units can result in small gaps and/or overlaps remaining between ``repaired" units.

\begin{figure}[h!]
\begin{center}
\includegraphics[height=1.5in]{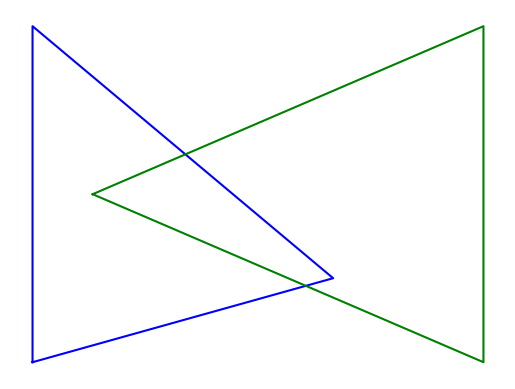}  \ \ \ 
\includegraphics[height=1.5in]{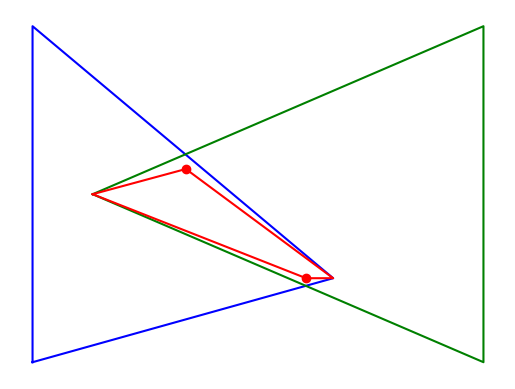}
\end{center}
\caption{Rounding error in computing overlap}
\label{intersection-example-fig}
\end{figure}

Additionally, the {\tt quick\_repair} algorithm only considers overlaps between pairs of units---but higher-order overlaps can and do occur in practice, as shown in Figure \ref{triple-overlap-fig}.  The triple overlap in the center is contained in each of the three pairwise overlaps, and it is guaranteed to be reassigned to at least two of the three intersecting units as the pairwise overlaps are reassigned.
\begin{figure}[h!]
\begin{center}
\includegraphics[height=1.5in]{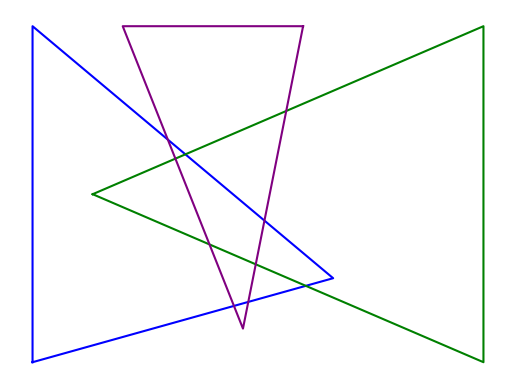}  
\end{center}
\caption{A triple overlap}
\label{triple-overlap-fig}
\end{figure}

The {\tt smart\_repair} algorithm avoids both of these problems by taking a completely different approach. First, we construct the simplicial 1-complex consisting of the
{\em fully noded union} of the unit boundaries $\{\partial P_1, \ldots, \partial P_N\}$.  In this construction, any line segment $\overline{pq}$ that intersects another line segment at some point $r$ in its interior is replaced by the pair of segments $\overline{pr}, \overline{rq}$. 
For example, the points of intersection between the units in Figure \ref{intersection-example-fig} would be added to the original unit boundaries as new vertices, resulting in the 1-complex shown in Figure \ref{boundaries-union-fig}(a).
\begin{figure}[h!]
\begin{center}
\subfloat[]{\includegraphics[height=1.5in]{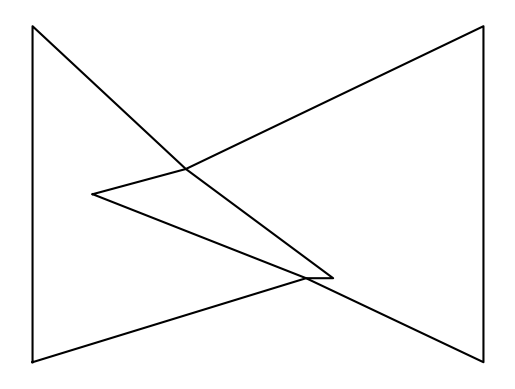}  }
\ \ \ 
\subfloat[]{\includegraphics[height=1.5in]{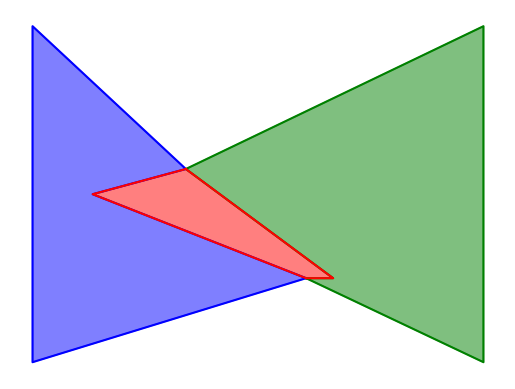}  }
\end{center}
\caption{(a) Fully noded union of boundaries of polygons in Figure \ref{intersection-example-fig}; (b) Polygonization of the 1-skeleton}
\label{boundaries-union-fig}
\end{figure}

Next, we ``polygonize" this 1-complex as in Figure \ref{boundaries-union-fig}(b) to create a clean partition of $R$ into polygons that intersect only along their boundaries; we will refer to the polygons in this partition as {\em pieces}. 
To each piece $Q$ in this partition, we associate the set 
\[ S_Q = \{k \in \{1,\ldots, N\} : Q \subset P_k\}; \]
i.e., $S_Q$ is the maximal subset of $\{1,\ldots, N\}$ for which $Q \subset \bigcap_{k \in S_Q} P_k$.
We refer to the cardinality of $S_Q$ as the {\em overlap order} of the piece $Q$. So, e.g., pieces of overlap order 1 each belong to exactly one unit, and pieces of overlap order 0 represent gaps.

\subsection{Step 2: Assign overlaps}\label{reconstruction-sec}
The first step in the construction of the repaired units $\widetilde{P}_1, \ldots, \widetilde{P}_N$ is to assign each piece $Q$ of overlap order 1 to the unique unit $\widetilde{P}_k$ for which $Q \subset P_k$.  Then we check each (partially) reconstructed unit $\widetilde{P}_k$ for connectedness; we will refer to any unit that has more than one connected component at this stage as ``disconnected."  

Next, starting with the order 2 overlaps:
\begin{enumerate}
\item For each disconnected unit, identify any overlaps contained in the corresponding original unit and assign them to this unit.  If this suffices to restore the unit to connectedness, remove it from the list of disconnected units.

Note that this process is not guaranteed to reconnect all disconnected units. 
There exist configurations in which the removal of a single overlap disconnects more than one unit, and no possible choice for how to assign this overlap would reconnect all units; see Figure \ref{bad-connectedness-fig} for an example.  In this case, some of the repaired units $\widetilde{P}_1, \ldots, \widetilde{P}_N$ may have multiple connected components. 
\item Assign each of the remaining overlaps to the unit with which its boundary shares the largest perimeter.
\end{enumerate}
Repeat this process for overlaps of order 3, etc., until all overlaps have been assigned to units.
\begin{figure}[h!]
\begin{center}
\includegraphics[height=1.5in]{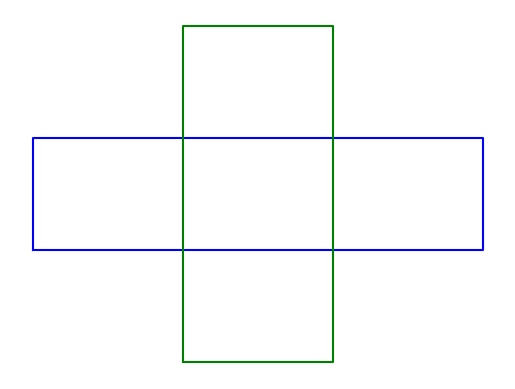}  
\end{center}
\caption{An overlap that disconnects multiple units}
\label{bad-connectedness-fig}
\end{figure}

\begin{example}\label{reconstruction-toy-example}
Consider a region consisting of the three unit polygons shown in Figure \ref{reconstruction-toy-example-fig}(a).  There are three overlaps of order 1, three overlaps of order 2, and one overlap of order 3.  Assigning overlaps  according to the algorithm described above produces the repaired region shown in Figure \ref{reconstruction-toy-example-fig}(e).
\begin{figure}[h!]
\begin{center}
\subfloat[]{\includegraphics[height=0.8in]{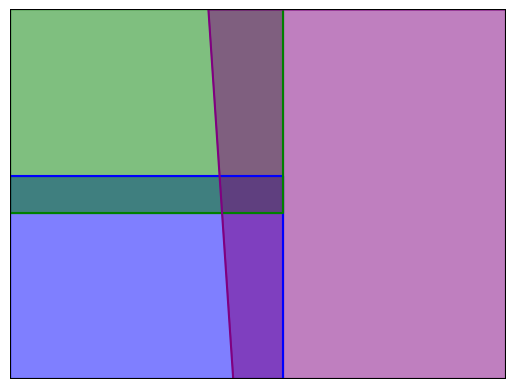} }
\hfill
\subfloat[]{\includegraphics[height=0.8in]{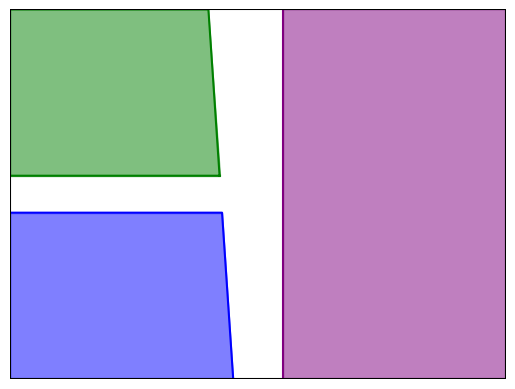}  }
\hfill
\subfloat[]{\includegraphics[height=0.8in]{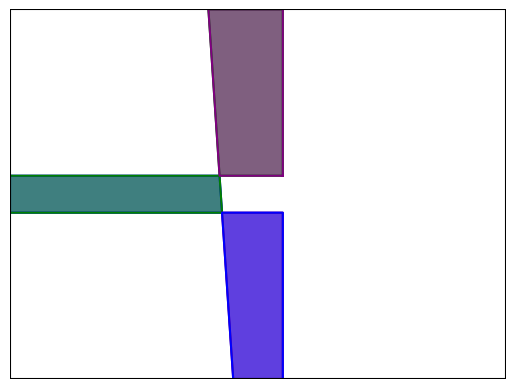}  }
\hfill
\subfloat[]{\includegraphics[height=0.8in]{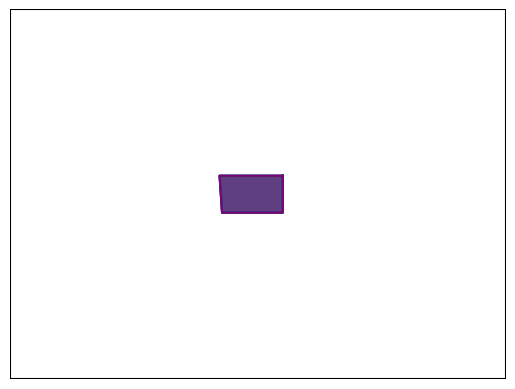}  }
\hfill
\subfloat[]{\includegraphics[height=0.8in]{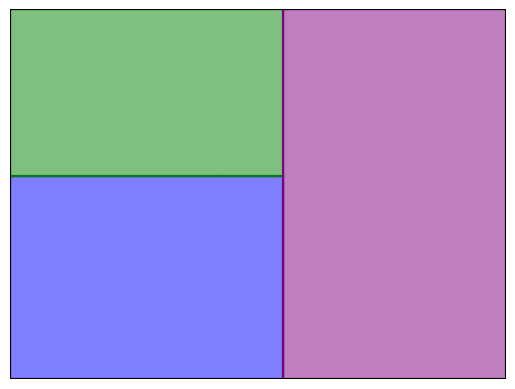} }
\end{center}
\caption{(a) Overlapping unit polygons; (b) Order 1 overlaps; (c) Order 2 overlaps; (d) Order 3 overlap; (e) Repaired unit polygons}
\label{reconstruction-toy-example-fig}
\end{figure}

\end{example}

\begin{example}\label{reconstruction-real-example} Figure \ref{reconstruction-real-example-fig}(a) shows three precincts from Colorado's 2020 precinct map; the blue precinct in the center is in Arapahoe County, while the green precinct surrounding it and the purple precinct to its south are in Denver County.  
There is considerable overlap between the Arapahoe County precinct and the Denver County precinct that surrounds it; 
the overlapping regions are shown in Figure \ref{reconstruction-real-example-fig}(b).

In particular, there is an overlapping region between these two precincts that extends southward to the northern boundary of the other Denver County precinct.  The {\tt quick\_repair} algorithm would assign this overlap to the Arapahoe County precinct (because its boundary shares a much larger perimeter with that precinct than with the Denver County precinct), thereby disconnecting the Denver precinct, as in Figure \ref{reconstruction-real-example-fig}(c). Because the {\tt smart\_repair} algorithm prioritizes connectivity over shared perimeter, it instead assigns this overlap to the Denver County precinct, as in Figure \ref{reconstruction-real-example-fig}(d).

\begin{figure}[h!]
\begin{center}
\subfloat[]{\includegraphics[height=1.1in]{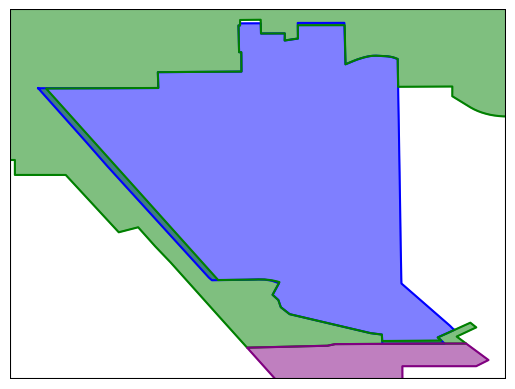} }
\ \ \ 
\subfloat[]{\includegraphics[height=1.1in]{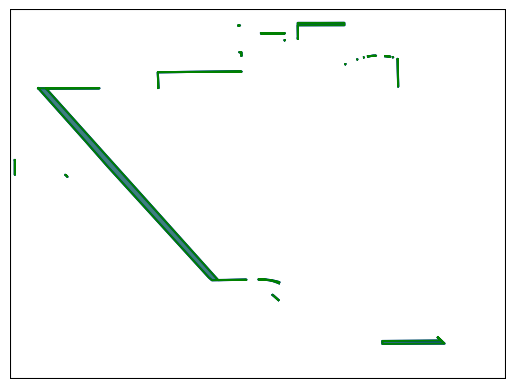}  }
\ \ \ 
\subfloat[]{\includegraphics[height=1.1in]{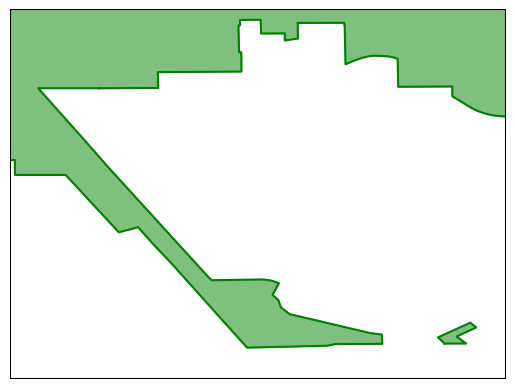}  }
\ \ \ 
\subfloat[]{\includegraphics[height=1.1in]{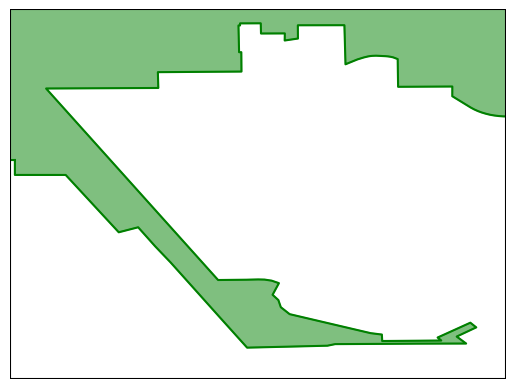} }
\end{center}
\caption{(a) Overlapping precincts; 
(b) Overlapping regions; (c)  Repair via {\tt quick\_repair} (d) Repair via {\tt smart\_repair}}
\label{reconstruction-real-example-fig}
\end{figure}

\end{example}

\subsection{Step 3: Close gaps}
For this step, consideration of adjacency relations motivates a complete departure from the {\tt quick\_repair} algorithm.  As the example described in Section \ref{repair-gone-wrong-subsec} vividly illustrates, whenever a gap is adjacent to more than a few units, assigning the entire gap to any single unit is practically guaranteed to create erroneous adjacency relations in the resulting ``repaired" tiling. Instead, the gap needs to be subdivided into smaller pieces, and different pieces should be assigned to different units in a way that results in reasonable adjacency relations between units.

Importantly, this algorithm is primarily intended to fill relatively small gaps between units.  For larger, more complicated gaps, the choices made by the algorithm for how to subdivide gaps may lead to unsatisfactory results.  Additionally, some larger gaps are intentional; e.g., some gaps representing large lakes may be not be included in any voting precinct and should ideally remain unfilled.\footnote{For this reason, the {\tt Maup} implementation of {\tt smart\_repair} includes a user-specified parameter (set at 0.1 by default) for which gaps whose area exceeds that fraction of the area of the largest adjacent unit will remain unfilled.}

The first step is to subdivide the boundary of the gap into its intersections of positive length with the boundaries of individual adjacent unit polygons; we will refer to these as {\em sub-boundaries} of the gap.  
Note that it is possible for the intersection of a gap boundary with a unit polygon to have multiple connected components; in this case we will consider the components as separate sub-boundaries.
For each gap, the algorithm proceeds based on the number of sub-boundaries that the gap boundary contains.

\subsubsection{Gaps with 1 sub-boundary}
The simplest case
is when the entire gap boundary is adjacent to a single unit.
In this case, the entire gap is contained within this unit, and we assign the gap to this unit.

\subsubsection{Gaps with 2 sub-boundaries}
For gaps with 2 sub-boundaries, any choice for how to divide the gap between the two adjacent units (e.g., assign the entire gap to the unit with which it shares the largest perimeter as in the {\tt quick\_repair} algorithm, or somehow divide the gap between the two units) will have an equivalent effect on the adjacency relations: These two units will be adjacent in the repaired file, and no other adjacencies between units will be affected by filling this gap.  

However, there are more geometric considerations that may make it desirable to divide the gap between the two unit polygons.  
In the absence of any ground truth information about the ``correct" way to fill gaps, we will need to adopt some guiding principles.  For gaps with 2 sub-boundaries, the following will suffice:

\begin{principle}\label{convexity-principle}
Optimize the convexity of the repaired unit polygons.
\end{principle}

For gaps with 2 sub-boundaries, we accomplish this by constructing the shortest path within the gap between the endpoints of the sub-boundaries, dividing the gap along this path, and assigning the resulting regions to their adjacent units.  Since the vertices in the shortest path are guaranteed to be vertices of the gap boundary (cf.~ Lemma \ref{shortest-path-lemma}), this process does not introduce any new vertices or points of intersection between unit boundaries.\footnote{
Because the vast majority of gaps encountered in practice are simply connected and shortest paths within non-simply connected polygons are not guaranteed to be unique, the {\tt Maup} implementation of {\tt smart\_repair} is restricted to simply connected gaps.  Any non-simply connected gaps (e.g., large bodies of water containing islands) are left unfilled and the user is notified of their presence.  (In practice, non-simply connected gaps are often large enough to remain unfilled by default in any case.)
}
An example of this process 
is shown in Figure \ref{two-sub-boundary-gap-example-fig}.
\begin{figure}[h!]
\begin{center}
\subfloat[]{\includegraphics[height=1.1in]{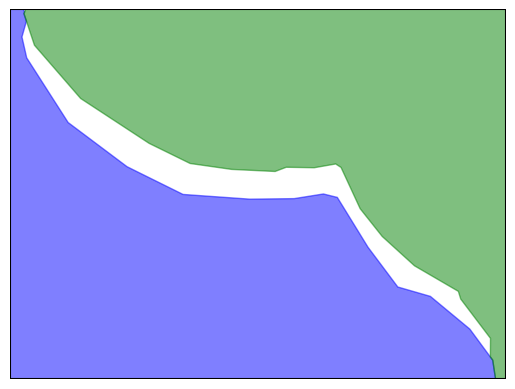}  }
\ \ \ 
\subfloat[]{\includegraphics[height=1.1in]{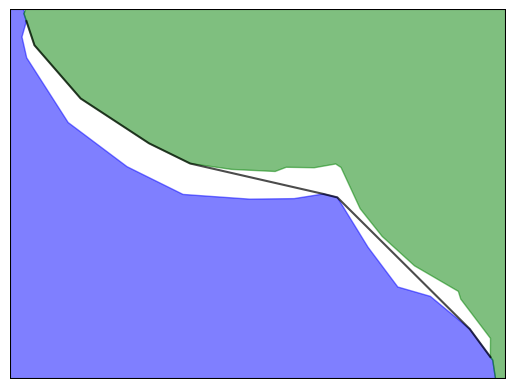}  }
\ \ \ 
\subfloat[]{\includegraphics[height=1.1in]{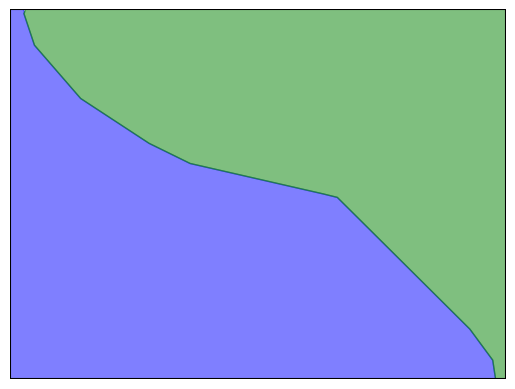} }
\ \ \ 
\subfloat[]{\includegraphics[height=1.1in]{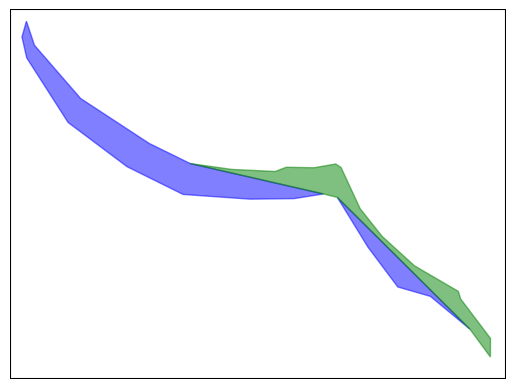}  }
\end{center}
\caption{(a) A gap with 2 sub-boundaries; (b) Shortest path between sub-boundary endpoints within gap; (c) Repaired unit polygons; (d) Filled gap}
\label{two-sub-boundary-gap-example-fig}
\end{figure}

\subsubsection{Gaps with 3 sub-boundaries}
Most gaps with 3 sub-boundaries have 3 adjacent units that are already pairwise adjacent to each other; thus they are similar to gaps with 2 sub-boundaries in the sense that any choice of how to divide the gap between the adjacent units will have an equivalent effect on the adjacency relations.

\begin{remark}
The exception to this scenario occurs when the boundary of some other unit intersects the gap boundary at one or more isolated points. The {\tt smart\_repair} algorithm is not designed to detect and accommodate such intersections, and in such cases the gap-filling procedure may create a ``false" adjacency between a pair of the gap's adjacent units.  Generally such ``false" adjacencies will have fairly short perimeter, and they can be removed with the optional rook-to-queen adjacency conversion step at the end of the repair algorithm.
\end{remark}

First, suppose that the gap is a simple triangle and that each of the gap's sub-boundaries is a line segment.  Such triangle-shaped gaps are often very long and thin, in which case assigning the entire gap to the unit with which it shares the largest perimeter generally produces a perfectly satisfactory result.  But for triangles that are closer to equilateral, assigning the entire triangle to any one unit creates a ``spike" on that unit polygon that may be undesirable from a convexity standpoint.  

One strategy for optimizing the convexity of all three unit polygons adjacent to the gap might be to 
divide the triangle along line segments between its vertices and its centroid, and assign each of the resulting regions to the unit adjacent to it, as in Figure \ref{sample-triangle-fig}(a)-(c).  This works well for triangles that are close to equilateral, but for long, thin triangles with one short side as in Figure \ref{sample-triangle-fig}(d), it still creates a significant spike on the unit polygon adjacent to the short side.  To mitigate this effect, we modify this strategy by using the {\em incenter} of the triangle---i.e., the common intersection point of the triangle's interior angle bisectors---instead of the centroid.  For triangles that are close to equilateral, the effect of this modification is minimal; for long, thin triangles, it has a larger effect and produces less dramatic spikes; see Figure \ref{sample-triangle-fig}(e)-(f).
\begin{figure}[h!]
\begin{center}
\subfloat[]{\includegraphics[height=1.3in]{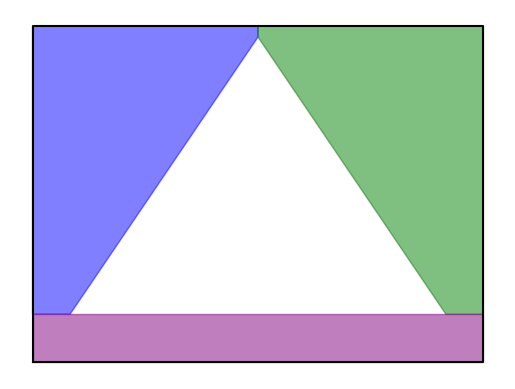} }
\ \ \ 
\subfloat[]{\includegraphics[height=1.3in]{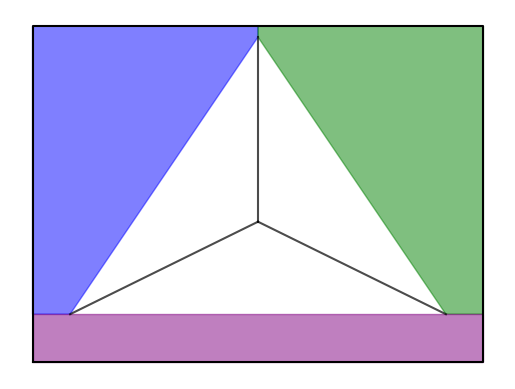} }
\ \ \ 
\subfloat[]{\includegraphics[height=1.3in]{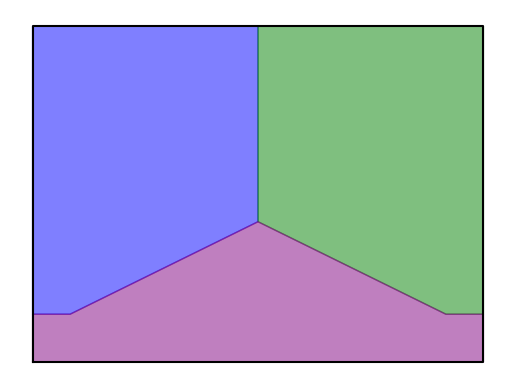}  }
\\ 
\subfloat[]{\includegraphics[height=1.3in]{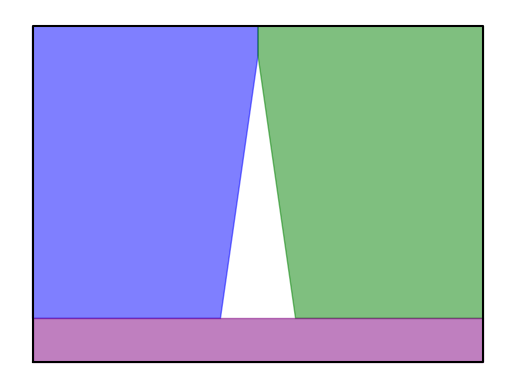}  }
\ \ \ 
\subfloat[]{\includegraphics[height=1.3in]{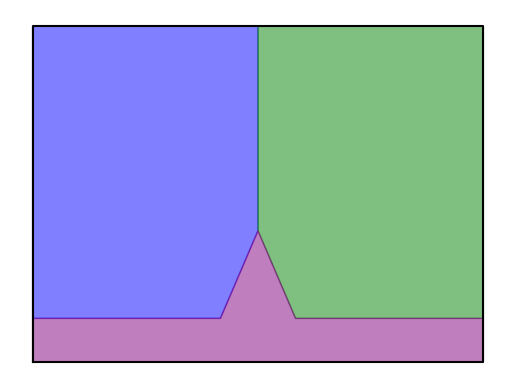}  }
\ \ \ 
\subfloat[]{\includegraphics[height=1.3in]{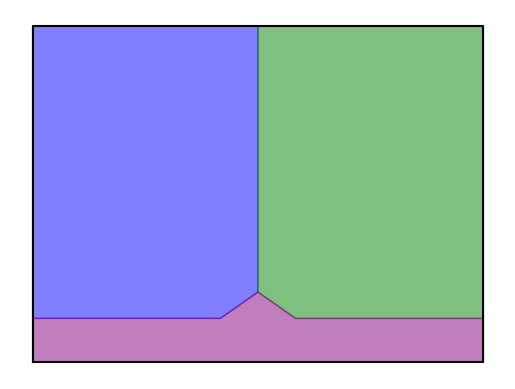} }
\end{center}
\caption{(a) ``Fat" triangle gap; (b)-(c) ``Fat" triangle gap filled by partitioning from centroid; (d) ``Thin" triangle gap; (e) ``Thin" triangle gap filled by partitioning from centroid; (f) ``Thin" triangle gap filled by partitioning from incenter }
\label{sample-triangle-fig}
\end{figure}

Now consider a more general gap with 3 sub-boundaries.  Polygons representing such gaps are often non-convex and can have surprisingly complicated geometry.  Some of this complexity can be reduced by applying the same convexity-optimizing strategy that we used for gaps with 2 sub-boundaries: Construct the shortest path within the gap between the endpoints of each sub-boundary, and assign the region bounded by each sub-boundary together with the shortest path between its endpoints to the unit adjacent to the sub-boundary.

A straightforward application of the triangle inequality shows that the shortest paths between endpoints of different sub-boundaries cannot cross each other---although they may intersect at a vertex or along a polygonal path---and therefore the regions created in this way are guaranteed to intersect only along their boundaries (if at all), and only along entire boundary line segments and/or at isolated vertices.  Therefore, this process does not create any overlaps or introduce any new vertices or points of intersection between unit polygon boundaries.  We will refer to this process as {\em convexification} of the gap sub-boundaries; see Figure \ref{three-sub-boundary-gap-convexification-example-fig} for an example.

\begin{figure}[h!]
\begin{center}
\subfloat[]{\includegraphics[height=1.3in]{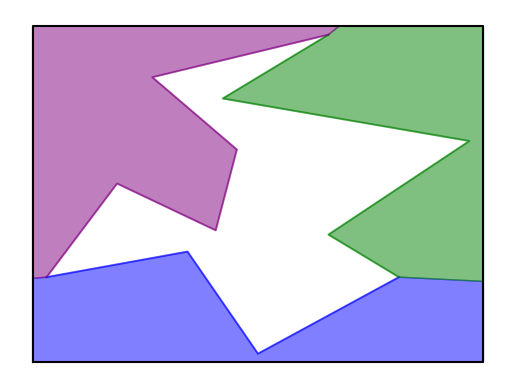}  }
\ \ \
\subfloat[]{\includegraphics[height=1.3in]{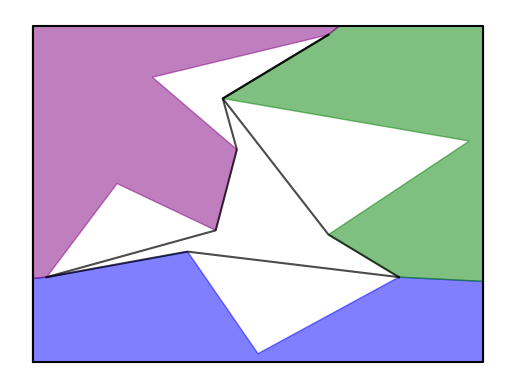}  }
\ \ \ 
\subfloat[]{\includegraphics[height=1.3in]{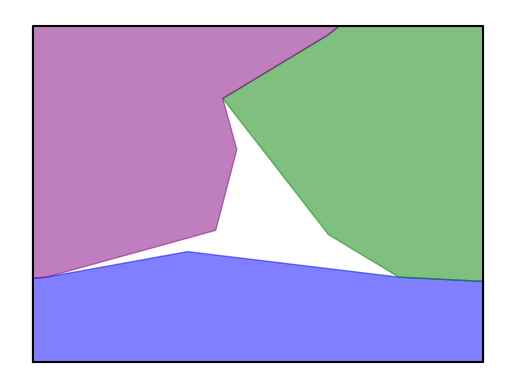} }
\end{center}
\caption{(a) A gap with 3 sub-boundaries; (b) Shortest paths within gap between sub-boundary endpoints; (c) Simplified unit polygons and simplified gap}
\label{three-sub-boundary-gap-convexification-example-fig}
\end{figure}

Unlike in the case of gaps with 2 sub-boundaries, when there are 3 sub-boundaries convexification generally does not completely fill the gap.  It does, however, leave a smaller gap with 3 sub-boundaries and a simpler geometric structure: Proposition \ref{convexification-prop} implies that after a gap has been simplified by convexifying its sub-boundaries, the new sub-boundaries of the simplified gap are all outward convex.
This implies that the convex hull of the gap is a triangle whose vertices are the endpoints of the gap sub-boundaries.  Let $q$ be the incenter of this triangle; then there are two possibilities:

\begin{enumerate}
\item If $q$ is contained in the interior of the gap, construct the shortest paths within the gap from $q$ to each of the endpoints of the sub-boundaries.  These paths partition the gap into regions that can each be assigned to their adjacent unit, similarly to the procedure for a simple triangle.  (See Figure \ref{three-sub-boundary-gap-fill-example-1-fig} for an example.)

\item If $q$ is not contained in the interior of the gap, then it is contained in a polygon bounded by exactly one of the gap sub-boundaries and the boundary of the gap's convex hull.  Without loss of generality, call this sub-boundary $B_1$ and the other two sub-boundaries $B_2$ and $B_3$.  Let $v_0$ be the gap vertex where $B_2$ and $B_3$ intersect, and let $v_1$ be the interior vertex of $B_1$ that is closest to $v_0$.  (Note that this scenario implies that $B_1$ contains at least one interior vertex in addition to its endpoints.)  Construct the shortest path within the gap between $v_0$ and $v_1$; this path partitions the gap into two regions, one adjacent to each of $B_2$ and $B_3$, which can then be assigned accordingly.  (See Figure \ref{three-sub-boundary-gap-fill-example-2-fig} for an example.)

\end{enumerate}

\begin{figure}[h!]
\begin{center}
\subfloat[]{\includegraphics[height=1.3in]{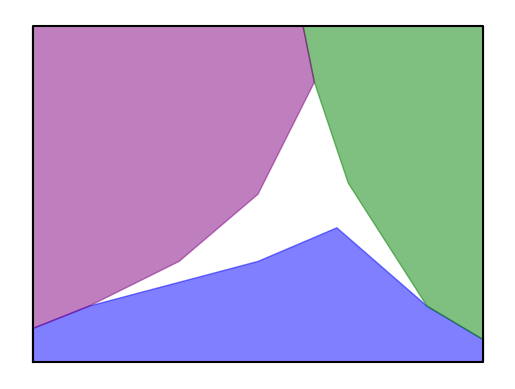}  }
\ \ \
\subfloat[]{\includegraphics[height=1.3in]{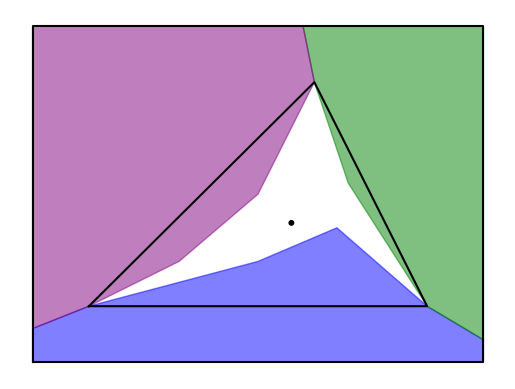}  }
\ \ \ 
\subfloat[]{\includegraphics[height=1.3in]{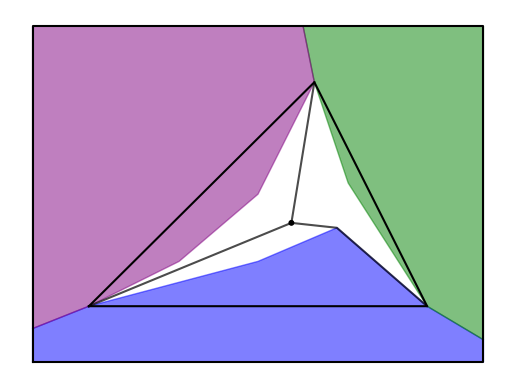} }
\ \ \ 
\subfloat[]{\includegraphics[height=1.3in]{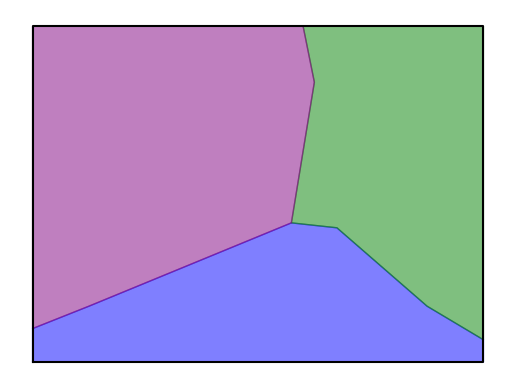} }
\ \ \ 
\subfloat[]{\includegraphics[height=1.3in]{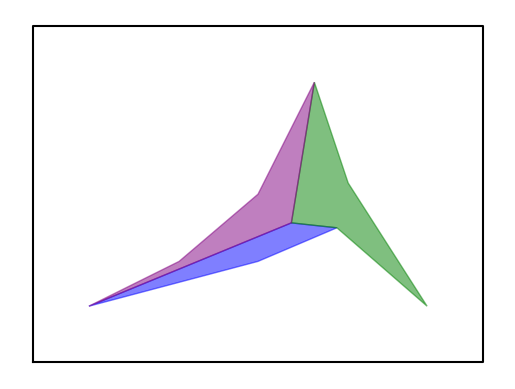} }
\end{center}
\caption{(a) A gap with 3 convexified sub-boundaries; (b) Convex hull of the gap, with incenter contained in the interior of the gap; (c) Shortest paths within gap between incenter and sub-boundary endpoints; (d) Repaired unit polygons; (e) Filled gap}
\label{three-sub-boundary-gap-fill-example-1-fig}
\end{figure}

\begin{figure}[h!]
\begin{center}
\subfloat[]{\includegraphics[height=1.3in]{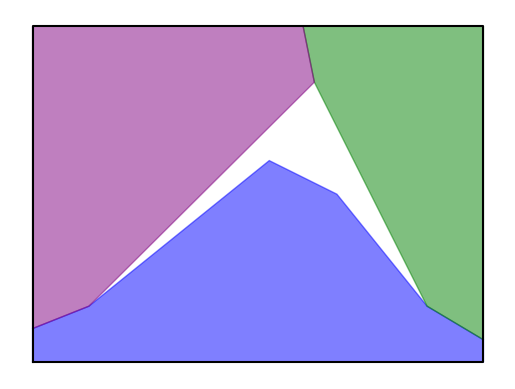}  }
\ \ \
\subfloat[]{\includegraphics[height=1.3in]{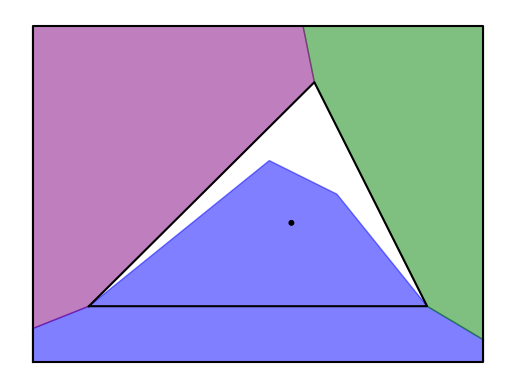}  }
\ \ \ 
\subfloat[]{\includegraphics[height=1.3in]{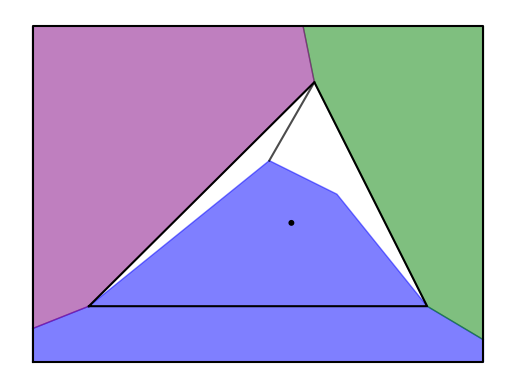} }
\ \ \ 
\subfloat[]{\includegraphics[height=1.3in]{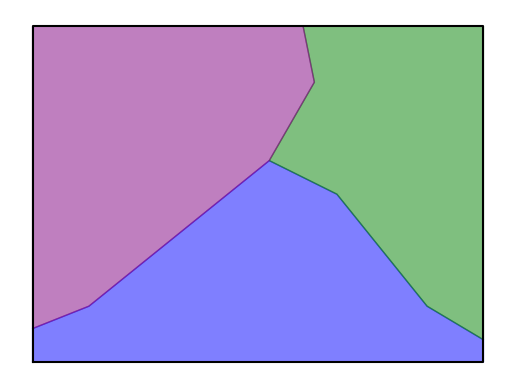} }
\ \ \ 
\subfloat[]{\includegraphics[height=1.3in]{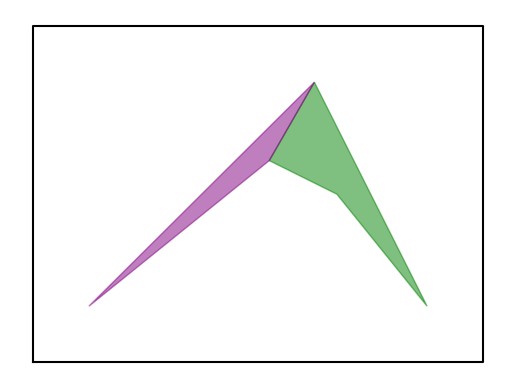} }
\end{center}
\caption{(a) A gap with 3 convexified sub-boundaries; (b) Convex hull of the gap, with incenter not contained in the interior of the gap; (c) Shortest path within gap between opposite vertex and nearest point on sub-boundary; (d) Repaired unit polygons; (e) Filled gap}
\label{three-sub-boundary-gap-fill-example-2-fig}
\end{figure}

\subsubsection{Gaps with 4 or more sub-boundaries}\label{four-gap-subsec}
Gaps with 4 or more sub-boundaries require more careful attention, because different choices for how to fill them will result in different adjacency relations between the units adjacent to the gap.  As for gaps with 3 sub-boundaries, it is usually the case that any two adjacent sub-boundaries have adjacent unit polygons that are already adjacent to each other, but choices must be made regarding which pairs of non-adjacent sub-boundaries are adjacent to unit polygons that should become adjacent once the gap is filled.

As for gaps with 3 sub-boundaries, the first step is to convexify the sub-boundaries by constructing the shortest path within the gap between the endpoints of each sub-boundary and assigning the region bounded by the sub-boundary together with this path to the unit adjacent to the sub-boundary.  In some cases this has the effect of dividing the remaining gap into two or more disjoint gaps. It is also possible that this will suffice to reduce the number of remaining sub-boundaries in one or more of the remaining gaps to 3 or fewer, in which case any such gaps can then be filled as in one of the previous cases.  

For the general case where the simplified gap still has at least 4 sub-boundaries, we adopt our next guiding principle:

\begin{principle}\label{nearest-pair-principle}
Consider all non-adjacent pairs of convexified sub-boundaries $(B', B'')$ of a simplified gap polygon $P$ that are strongly mutually visible in $P$.  Among all such pairs, the units adjacent to the pair with the shortest distance between them should become adjacent after the gap is repaired.
\end{principle}

With this principle in mind, we proceed with the simplified gap as follows:
\begin{enumerate}
\item Compute the distance between each non-adjacent pair of convexified sub-boundaries, and choose the non-adjacent pair for which this distance is shortest; call these sub-boundaries $B'$ and $B''$.  
\item With all sub-boundaries oriented counter-clockwise around the boundary of the simplified gap polygon, construct the shortest path $\alpha_1$ within the gap polygon from the terminal point of $B'$ to the initial point of $B''$, and the shortest path $\alpha_2$ within the gap polygon from the terminal point of $B''$ to the initial point of $B'$. 
\begin{itemize}
\item If $\alpha_1$ and $\alpha_2$ intersect, then Theorem \ref{smv-criterion-thm} implies that  $B'$ and $B''$ are not strongly mutually visible, so go back to the previous step and select the non-adjacent pair of sub-boundaries with the next-shortest distance between them.
\item If $\alpha_1$ and $\alpha_2$ are disjoint, then Theorem \ref{smv-criterion-thm} implies that  $B'$ and $B''$ are strongly mutually visible, so go on to the next step in order to create an adjacency between the unit polygons adjacent to $B'$ and $B''$.
\end{itemize}

\item
Construct the shortest path $\beta_1$ within the gap between the initial points of $B'$ and $B''$ and the shortest path $\beta_2$ within the gap between the terminal points of $B'$ and $B''$.  By Theorem \ref{positive-area-thm}, the union of $B'$, $B''$, $\beta_1$, and $\beta_2$ bounds either one or two simple polygons, and each of these polygons intersects exactly one of $B'$, $B''$ in a path of positive length.
\item Assign each polygon created in the previous step to the unit polygon adjacent to either $B'$ or $B''$ with which it shares a boundary of positive length.  According to Theorem \ref{positive-area-thm}, the resulting unit polygons will now intersect at a single common vertex.
Furthermore, this will leave either one or two smaller gaps remaining, each of which contains strictly fewer sub-boundaries than the original gap.  
\item Repeat the entire process for each of these smaller gaps.  Since the number of gap sub-boundaries in each gap decreases with each iteration, this process will eventually fill all gaps completely.
\end{enumerate}

Note that in Step (4) above, we only created a single point adjacency between the unit polygons adjacent to $B'$ and $B''$, so this step does not yet create a common boundary of positive length between these units.  But in the smaller gaps created by this process, the new sub-boundaries corresponding to these units are adjacent, and at some subsequent point in the process these gaps will be filled by the algorithms described above for gaps with either 2 or 3 sub-boundaries, resulting in a positive-length boundary between these two units.

An example of the entire process applied to a gap with 4 sub-boundaries is shown in Figure \ref{four-sub-boundary-gap-example-fig}.

\begin{figure}[h!]
\begin{center}
\subfloat[]{\includegraphics[height=1.1in]{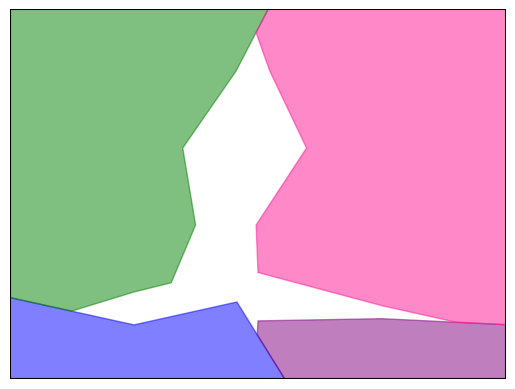}  }
\ \ \
\subfloat[]{\includegraphics[height=1.1in]{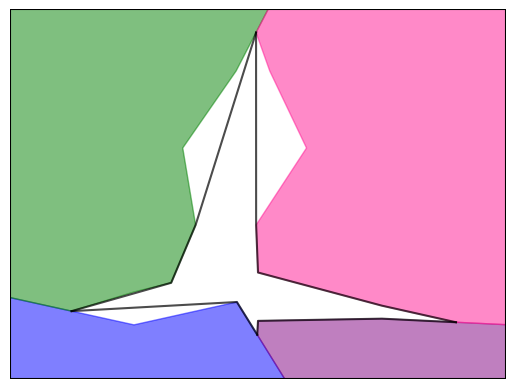}  }
\ \ \ 
\subfloat[]{\includegraphics[height=1.1in]{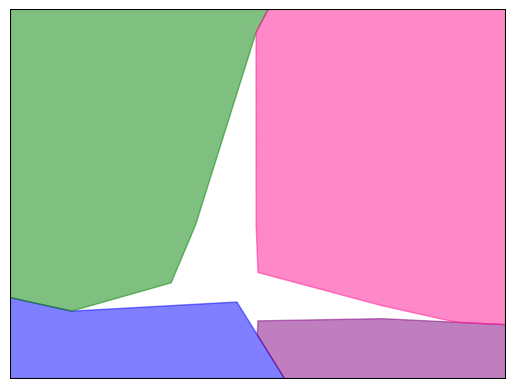} }
\ \ \ 
\subfloat[]{\includegraphics[height=1.1in]{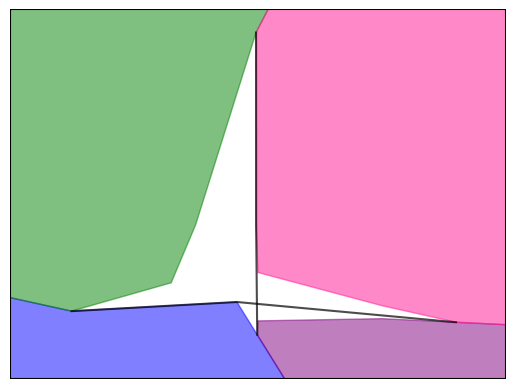}  }
\ \ \ 
\subfloat[]{\includegraphics[height=1.1in]{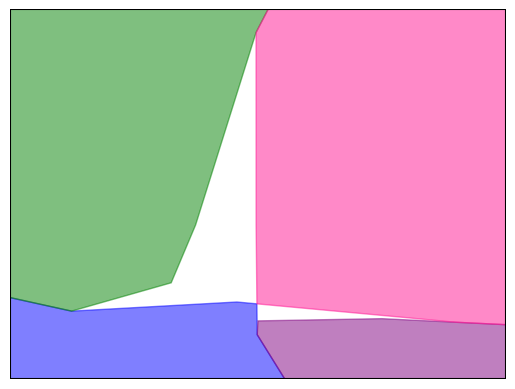} }
\ \ \ 
\subfloat[]{\includegraphics[height=1.1in]{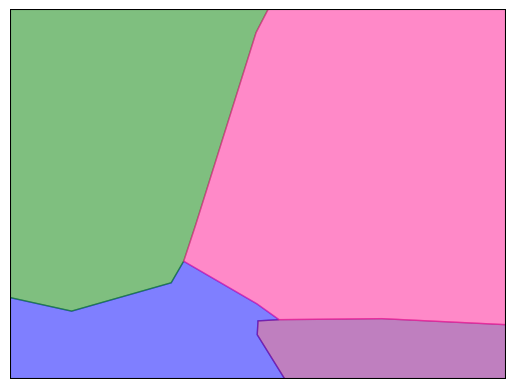} }
\ \ \ 
\subfloat[]{\includegraphics[height=1.1in]{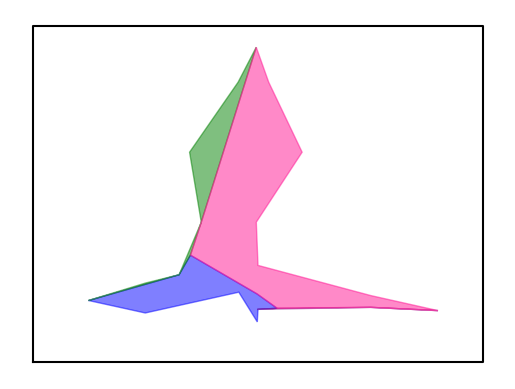} }
\end{center}
\caption{(a) A gap with 4 sub-boundaries; (b) Shortest paths between sub-boundary endpoints within gap; (c) Convexified unit polygons and simplified gap; 
(d) Shortest paths between initial and terminal points of closest non-adjacent strongly mutually visible sub-boundary pair; (e) Partially repaired unit polygons and two new gaps with 3 sub-boundaries each; (f) Completely repaired unit polygons; (g) Completely filled gap}
\label{four-sub-boundary-gap-example-fig}
\end{figure}

The following theorem confirms that this algorithm conforms with Guiding Principle \ref{nearest-pair-principle}.
It follows directly from Theorem \ref{smv-criterion-thm}, Theorem \ref{positive-area-thm}, and the adjacency structure that results from the algorithms for filling gaps with 2 and 3 sub-boundaries. 

\begin{theorem}\label{alg-adjacency-thm}
Consider a simplified gap polygon with at least 4 sub-boundaries, all of which are outward convex.  Among all strongly mutually visible pairs of non-adjacent sub-boundaries, let $(B', B'')$ be the pair with the shortest distance between them.  Then at the conclusion of the {\tt smart\_repair} algorithm, the unit polygons adjacent to $B'$ and $B''$ will share a boundary of positive length.
\end{theorem}

\subsection{Step 4: Cleaning up}
Unfortunately, configurations of overlapping unit polygons such as the one shown in Figure \ref{bad-connectedness-fig} are fairly common in practice, and at this point in the process there often remain some disconnected units.  Fortunately, in practice the extra ``orphaned" components usually represent a very small fraction of the area of the original unit polygon, and in this case it seems reasonable to resolve the connectivity issue by reassigning these small components to an adjacent unit.
Thus the final step in the main procedure is as follows:
\begin{enumerate}
\item Identify any units that are still disconnected, and sort their components by area, from smallest to largest.
\item If the smallest component has area less than a user-specified parameter\footnote{Set at 0.0001 by default in the {\tt Maup} implementation} times the area of the largest component, remove it from its assigned unit and instead assign it to the adjacent unit with which it shares the largest perimeter.
\item Repeat Step (2) until either the unit is no longer disconnected or the smallest component is larger than the specified threshold.\footnote{If any unit polygons remain disconnected at the end of this process, the {\tt Maup} implementation of {\tt smart\_repair} reports a list of disconnected units to alert the user that this has occurred.}
\end{enumerate}

\section{Optional features}\label{bells-and-whistles-sec}

\subsection{Nesting into larger units}
In applications such as redistricting, it is often necessary to combine information from different geographic units; for instance, population data from the decennial U.S. census is reported at the level of census blocks, while (in the best-case scenario) election results are reported at the level of voting precincts.  Aggregating and disaggregating data between units at different scales is required in order to collect all the necessary information on a single set of units.\footnote{The {\tt Maup} package contains a variety of functions intended to help with this task.}

Here we encounter another common problem: When two maps whose extent covers a common region (e.g., a U.S. city or state) are created by different agents, the total geographic extent covered by the unit polygons is often slightly---or not so slightly!---different between the two maps.  This can result in a variety of issues that complicate the aggregation/disaggregation of data between units; one common example is that some populated census blocks may have no intersection with any voting precinct, in which case the process of aggregating population data from blocks to precincts loses the data contained in those blocks.

Fortunately, sometimes we have additional information to guide the repair of a noisy map.  For example, in most U.S. states voting precincts are completely contained within counties---and clean, accurate maps of county boundaries within states are available from the U.S. Census Bureau.\footnote{https://www.census.gov}  When the units being repaired are intended to nest cleanly into some larger regions (e.g., counties), the {\tt smart\_repair} algorithm allows the user the option of providing a map of the region boundaries; it will then perform the repair so that the repaired units nest cleanly into the region boundaries and the total geographic extent of the repaired units agrees with the total geographic extent of the region map.  This is accomplished by making the following modifications to the algorithm described in Section \ref{main-alg-sec}:
\begin{itemize}
\item {\bf Assign units to regions:} Each of the unit polygons to be repaired is assigned to the region that it intersects with the largest area.  (Note that this option is only intended for situations in which each of the primary units is almost entirely contained within a single region and it is clear which region each unit should belong to.)
\item {\bf Construct refined tiling:} The region boundaries are included along with the main unit polygon boundaries in the construction of the 1-complex and its polygonization; this guarantees that each of the pieces in the refined tiling is fully contained within a single region.  When associating to each piece the set of unit polygons that contain it, we only include units that are assigned to the region containing that piece.  Additionally:
\begin{itemize}
\item  Any piece that is not contained in any region is dropped and will not be included in the process of assigning overlaps; this ensures that the total geographic extent of the repaired units will be contained within the total geographic extent of the region map.
\item ``Gaps" may now include pieces that are contained in some region but not contained within the total geographic extent of the primary units. Such gaps will have one or more {\em exterior} sub-boundaries in addition to sub-boundaries that are adjacent to unit polygons.  Filling these gaps ensures that---subject to the constraints described in Section \ref{main-alg-sec} that may leave some gaps unfilled---the total geographic extent of the repaired units will include the total geographic extent of the region map.

\end{itemize}

\item {\bf Assign overlaps:} The process of assigning overlaps is performed one region at a time, using only pieces contained in each region to reconstruct the unit polygons assigned to that region.

\item {\bf Close gaps:} The gap-closing step is also performed one region at a time.  The main new consideration in this case is that some gaps may have one or more sub-boundaries consisting of intersections with a region boundary.  For these exterior sub-boundaries, there is no adjacent unit polygon to assign any portion of the gap to, so we must partition the gap differently when it has one or more exterior sub-boundaries. This also means that we cannot apply our convexification procedure to simplify exterior sub-boundaries.  So after convexifying all the non-exterior sub-boundaries, we modify the gap-closing procedure as follows:
\begin{itemize}
\item If the gap has only one non-exterior sub-boundary, assign the entire gap to the unit adjacent to the non-exterior sub-boundary.
\item If the gap has 3 sub-boundaries, two of which are adjacent to distinct units and one of which is exterior, identify the exterior sub-boundary vertex that is closest to the vertex where the non-exterior sub-boundaries intersect.  Construct the shortest path within the gap between these two vertices, partition the gap along this path into two regions, and assign each region to its adjacent unit. (Note that we choose the closest {\em vertex} in the exterior sub-boundary rather than the closest {\em point} in order to avoid introducing any new points of intersection along boundaries that might be shared with unit polygons from other regions.)  An example of this process is shown in Figure \ref{three-sub-boundary-exterior-gap-example-fig}.
\item If the gap has 4 or more sub-boundaries (and at least two non-exterior sub-boundaries adjacent to distinct units), proceed as before unless the non-adjacent pair of sub-boundaries with the shortest distance between them includes an exterior sub-boundary.  In that case, identify the exterior sub-boundary vertex that is closest to the non-exterior sub-boundary, and construct the shortest paths within the gap from this vertex to each of the two endpoints of the non-exterior sub-boundary. (If both sub-boundaries in the non-adjacent pair are exterior, skip this pair and proceed to the pair with the next-shortest distance between them.)  Generically, these paths together with the non-exterior sub-boundary will bound a polygon that intersects the exterior boundary at a single vertex; we assign this polygon to the unit adjacent to the non-exterior sub-boundary.  An example of this process is shown in Figure \ref{four-sub-boundary-exterior-gap-example-fig}.
\end{itemize}

\end{itemize}

We will refer to this version of the {\tt smart\_repair} algorithm as the {\em region-aware} version.

\begin{figure}[h!]
\begin{center}
\subfloat[]{\includegraphics[height=1.1in]{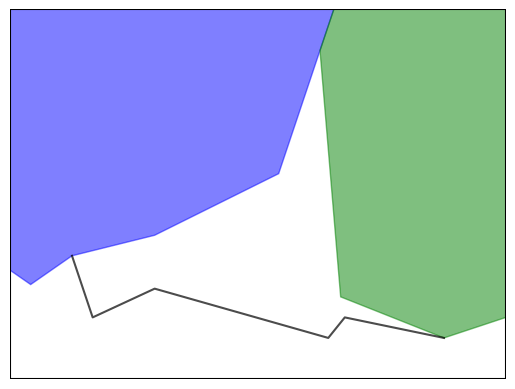}  }
\ \ \ 
\subfloat[]{\includegraphics[height=1.1in]{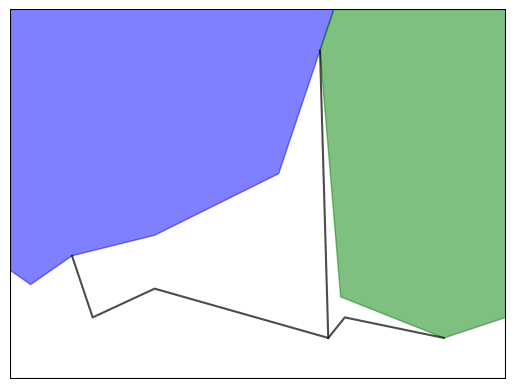}  }
\ \ \ 
\subfloat[]{\includegraphics[height=1.1in]{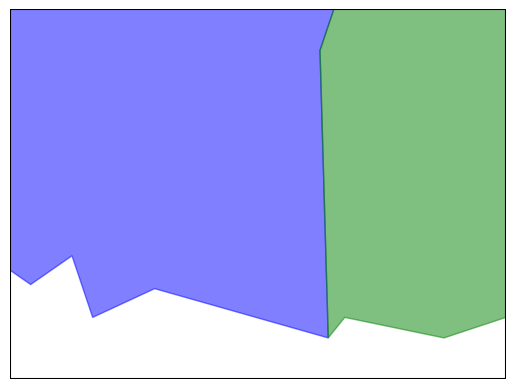} }
\ \ \ 
\subfloat[]{\includegraphics[height=1.1in]{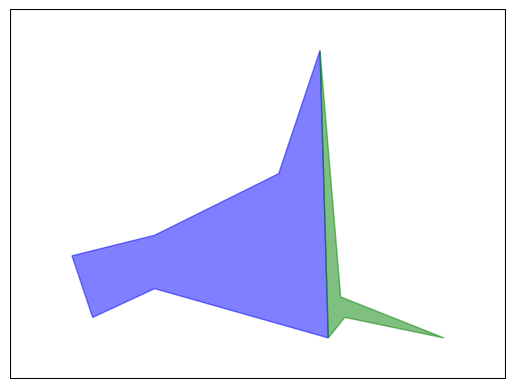} }
\end{center}
\caption{(a) A gap with 2 interior sub-boundaries and 1 exterior sub-boundary; (b) Shortest path between interior vertex and nearest exterior vertex; (c) Repaired unit polygons; (d) Filled gap}
\label{three-sub-boundary-exterior-gap-example-fig}
\end{figure}

\begin{figure}[h!]
\begin{center}
\subfloat[]{\includegraphics[height=1.1in]{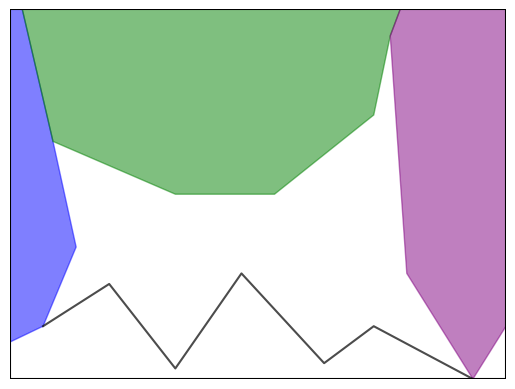}  }
\ \ \ 
\subfloat[]{\includegraphics[height=1.1in]{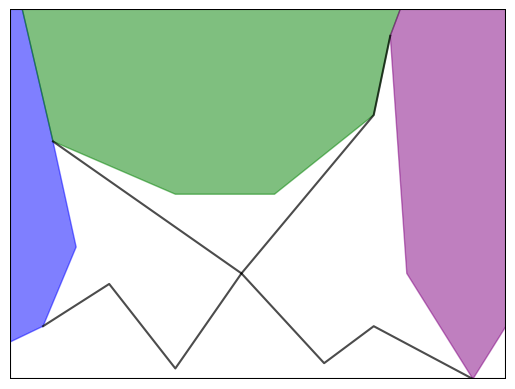}  }
\ \ \ 
\subfloat[]{\includegraphics[height=1.1in]{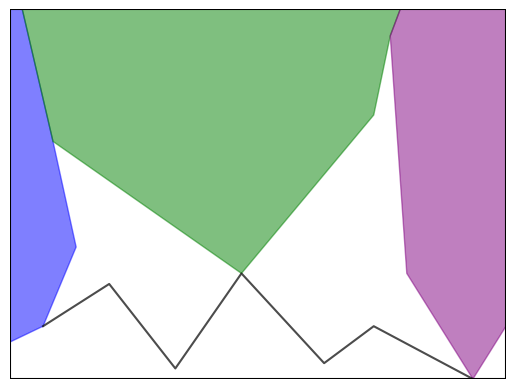} }
\ \ \ 
\subfloat[]{\includegraphics[height=1.1in]{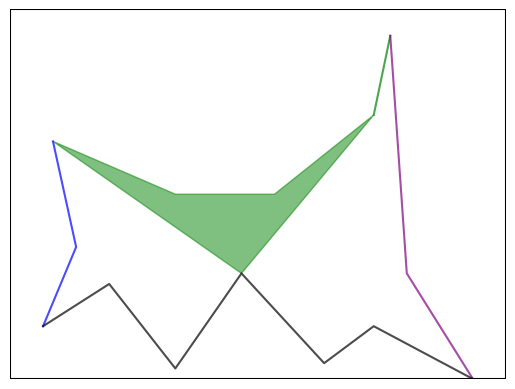} }
\end{center}
\caption{(a) A gap with 3 interior sub-boundaries and 1 exterior sub-boundary; (b) Shortest paths between endpoints of non-adjacent interior sub-boundary and nearest exterior vertex; (c) Partially repaired unit polygons; (d) Partially filled gap}
\label{four-sub-boundary-exterior-gap-example-fig}
\end{figure}

\begin{example}\label{toy-region-aware-example}
Suppose that the ``precincts" shown in Figure \ref{toy-region-aware-example-fig}(a) are intended to nest cleanly into the ``counties" shown in Figure \ref{toy-region-aware-example-fig}(b).  Figure \ref{toy-region-aware-example-fig}(c) shows the result of primary {\tt smart\_repair} algorithm, while Figure \ref{toy-region-aware-example-fig}(d) shows the result of the region-aware {\tt smart\_repair} algorithm.

\begin{figure}[h!]
\begin{center}
\subfloat[]{\includegraphics[height=2in]{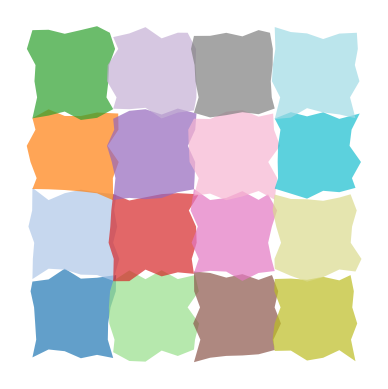} }
\ \ \ 
\subfloat[]{\includegraphics[height=2in]{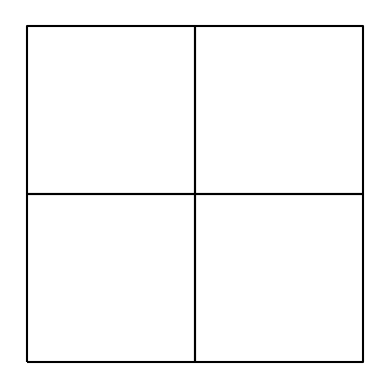} }
\\
\subfloat[]{\includegraphics[height=2in]{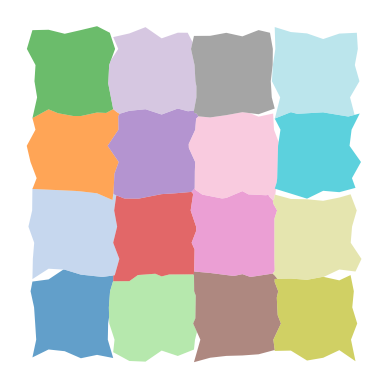} }
\ \ \ 
\subfloat[]{\includegraphics[height=2in]{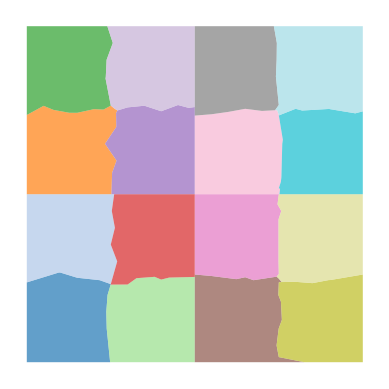} }
\end{center}
\caption{(a) ``Precincts" with gaps and overlaps; (b) ``Counties" that precincts should nest into; (c) Gaps and overlaps repaired without region-awareness; (d) Gaps and overlaps repaired with region-awareness}
\label{toy-region-aware-example-fig}
\end{figure}

\end{example}

\subsection{Small rook-to-queen adjacency conversion} 
Consider the example shown in in Figure \ref{four-sub-boundary-gap-example-fig}, where the {\tt smart\_repair} algorithm created a boundary of positive length between one of the two pairs of non-adjacent units.  Generically this is the right thing to do, but it is certainly plausible that these four units might have been intended to intersect in a common corner point, with neither of the non-adjacent pairs sharing a boundary of positive length.  In such a case, the length of the ``false" boundary created by the repair will typically be fairly small.  There also may be other such short, ``false" boundaries of positive length due to overlaps and other minor inaccuracies in the original unit polygons.  For applications such as redistricting where we want the adjacency relations between repaired units to be as accurate as possible, the {\tt smart\_repair} algorithm includes an optional step that converts all rook adjacencies with boundary length below a (typically very small) user-specified threshold to queen adjacencies.  This is accomplished as follows:
\begin{enumerate}
\item For each adjacency below the threshold length, construct a disk centered at the midpoint of the line segment between the adjacency's endpoints and with radius slightly larger than half the adjacency's length; this guarantees that the entire adjacency is contained within the interior of the disk.
\item Remove the disk by replacing each unit polygon that it intersects with the set difference of the unit polygon and the disk.\footnote{In order to avoid introducing small gaps and overlaps due to rounding errors during this process, the {\tt Maup} implementation of this step is similar to the first two steps of the primary algorithm: We construct the simplicial 1-complex consisting of the union of the boundary of the disk and the boundaries of all unit polygons intersecting the disk, polygonize this 1-complex, and reconstruct the unit polygons from the pieces that lie outside the disk.}
\item For each sub-boundary of the gap created by removing the disk, assign to the adjacent unit polygon a ``pie piece" bounded by the sub-boundary and line segments from its endpoints to the center of the disk. 
\end{enumerate}
An illustration of this process applied to the adjacency constructed in Figure \ref{four-sub-boundary-gap-example-fig} is shown in Figure \ref{small-rook-to-queen-example-fig}.
\begin{figure}[h!]
\begin{center}
\subfloat[]{\includegraphics[height=1.3in]{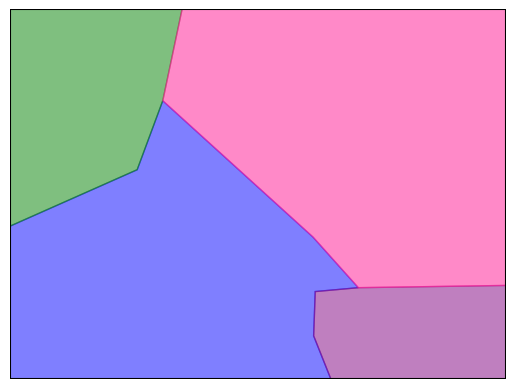} }
\ \ \ 
\subfloat[]{\includegraphics[height=1.3in]{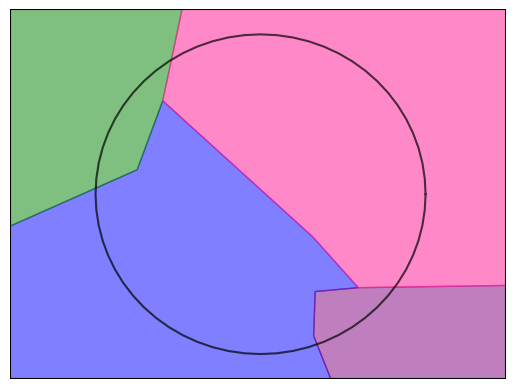} }
\ \ \ 
\subfloat[]{\includegraphics[height=1.3in]{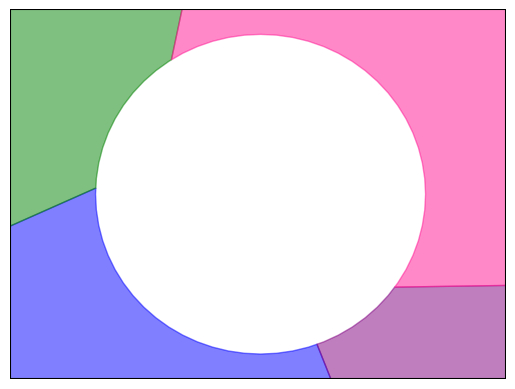} }
\ \ \ 
\subfloat[]{\includegraphics[height=1.3in]{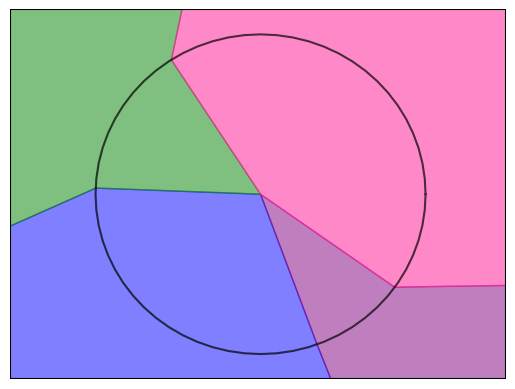} }
\ \ \ 
\subfloat[]{\includegraphics[height=1.3in]{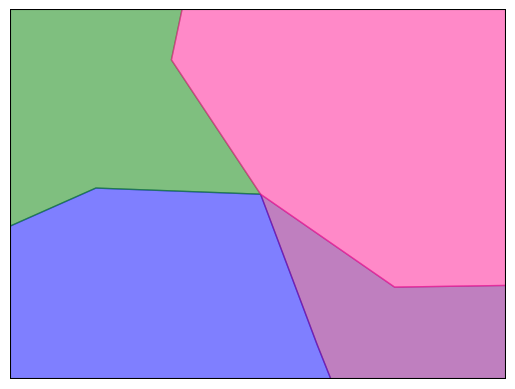} }
\end{center}
\caption{(a) A small rook adjacency; (b) Disk bounding adjacency; (c) Unit polygons with disk removed; (d) Disk replaced with ``pie pieces;" (e) Resulting queen adjacency}
\label{small-rook-to-queen-example-fig}
\end{figure}

In rare cases, it may happen that two or more of the disks constructed in Step (1) above intersect nontrivially.  In this case, we take the convex hull of the union of these disks and apply the remainder of the procedure to this polygon (with the polygon's centroid replacing the center of the disk in the previous construction), so that all of the small rook adjacencies within this polygon are converted to a single queen adjacency.  For example, when we applied this procedure to the Colorado 2020 voting precinct map with a length threshold of 100 feet, we found 188 adjacencies below the threshold, including two pairs for which the disks overlapped.  One of these ``double disks" and its intersection with 5 precincts, together with the result of the conversion procedure, is shown in Figure \ref{small-rook-to-queen-double-disk-example-fig}.
\begin{figure}[h!]
\begin{center}
\subfloat[]{\includegraphics[height=1.3in]{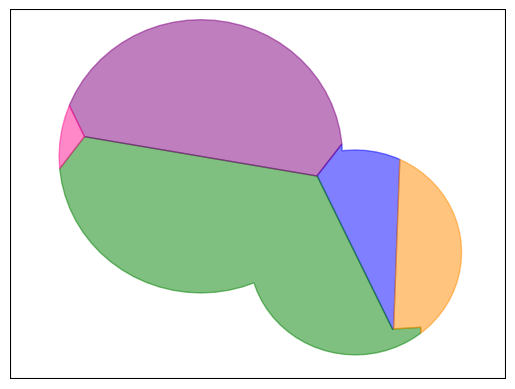} }
\ \ \ 
\subfloat[]{\includegraphics[height=1.3in]{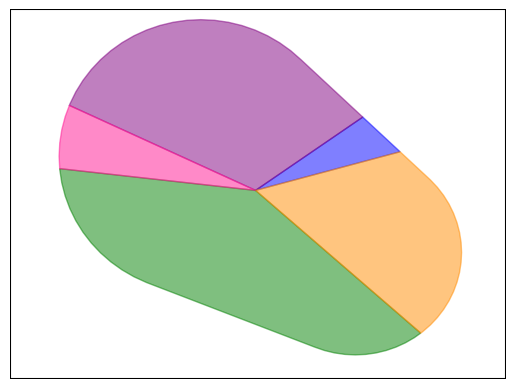} }
\end{center}
\caption{(a) A ``double disk" enclosing two small rook adjacencies; (b) Convex hull with both rook adjacencies replaced by a single queen adjacency}
\label{small-rook-to-queen-double-disk-example-fig}
\end{figure}

\section{Runtime complexity}\label{polynomial-time-sec}
  
In this section we perform a rough estimate for the runtime complexity for the {\tt smart\_repair} algorithm.  This estimate is not intended to be optimal, and it is not our primary concern, as the development of the algorithm was motivated primarily by a desire for accuracy rather than efficiency.  We have found that in practice, the {\tt Maup} implementation can repair most state-level precinct maps in a few hours on a standard laptop computer, which is more than adequate for our needs.

For purposes of this estimate, let $N$ denote the number of unit polygons in the approximate tiling $T$, and let $E$ denote the total number of edges in all the polygons in $T$.  Since each polygon must have at least 3 edges, we must have $N \leq \tfrac{1}{3} E$.

\noindent {\bf Construct refined tiling:}
\begin{itemize}
\item The fully noded union of the polygon boundaries is constructed by computing intersection points between all edges in all polygons.  This requires at most ${E \choose 2}$ operations (and in practice, generally much fewer than this\footnote{For instance, the Bentley-Ottmann algorithm of \cite{BO79} has run time complexity $O((E+k)\log E)$, where $k$ is the number of intersection points.}), and the number $E'$ of edges in the resulting fully noded union is bounded by $E' \leq 4{E \choose 2} = 2E(E-1)$.  
\item According to \cite{JB93}, the runtime complexity of the polygonization construction is at most $O(E' \log E')$, 
and the number $N'$ of distinct polygons produced is bounded by $N' \leq \tfrac{2}{3} E'$.
\item Computing the intersection of each piece $Q$ of the refined tiling  with each of the polygons in $T$ to compute the set $S_Q$ and its overlap order requires at most $NN'$ operations (again, generally much fewer than this in practice).
\end{itemize}

\noindent {\bf Assign overlaps:}
\begin{itemize}
\item Each piece of the refined tiling with overlap order 1 must be assigned to the unit polygon that it intersects; this requires at most $N'$ operations.
\item For each subsequent overlap order $d \geq 2$, the perimeter of the intersection of each piece $Q$ of the refined tiling of overlap order $d$ with each partially reconstructed unit polygon in $S_Q$ must be computed; then the maximum value of all these perimeters must be computed to decide which unit polygon to assign $Q$ to.  The worst-case scenario for the cardinality of $S_Q$ for any piece $Q$ is the number of unit polygons $N$, so this requires at most $N$ operations for each piece $Q$.  Since there are $N'$ pieces to consider, this requires at most $NN'$ operations in total.
\end{itemize}

\noindent {\bf Close gaps:}
\begin{itemize}
\item Gaps are represented by pieces of the refined tiling of overlap order 0.  For each gap, the intersection with each partially reconstructed unit polygon must be computed in order to identify the sub-boundaries of the gap.  This requires at most $NN'$ operations.
\item Convexification of each gap requires the computation of the shortest path between the endpoints of the gap sub-boundary.  According to \cite{LP84}, the runtime complexity for each shortest path construction in a gap with $n$ edges is at most $O(n \log n)$.  Since the total number of edges in all gaps is bounded by $E'$, it follows that the runtime complexity for each shortest path construction is at most $O(E' \log E')$.  The total number of sub-boundaries in all gaps is also bounded by $E'$, so the total runtime complexity for this step is at most $O((E')^2 \log E')$.
\item Remaining simplified gaps all have at least 3 sub-boundaries.  For each gap with $s \geq 4$ sub-boundaries, it generally requires $(s-3)$ applications of the procedure outlined in \S \ref{four-gap-subsec} to reduce the gap to $(s-2)$ smaller gaps with 3 sub-boundaries each.  Each application of this procedure  requires computing the distance between each of the non-adjacent pairs of sub-boundaries in the gap, and then construction of 4 shortest paths between endpoints of the pair with the shortest distance, each with runtime complexity bounded by $O(E' \log E')$.  Since the total number of sub-boundaries for all gaps is bounded by $E'$, the total run time complexity for this step is at most $O((E')^3)$ to compute all distances between pairs of sub-boundaries for all $(s-3)$ iterations, and $O((E')^2 \log E')$ for computing shortest paths, for an overall runtime complexity of $O((E')^3)$.
\item Remaining simplified gaps now all have 3 sub-boundaries each, and the time required to fill a gap with 3-sub-boundaries is independent of the number of edges in the gap.  The number of gaps remaining is bounded by $E'$, so the runtime complexity for this step is at most $O(E')$.
\end{itemize}

Putting it all together, we have the following rough estimate for the runtime complexity:

\begin{theorem}\label{run-time-theorem}
Let $T = \{P_1, \ldots, P_N\}$ be an approximate tiling of a simply connected region $R$ in the plane, as described in Problem \ref{big-problem}. Let $E$ denote the total number of edges in all the polygons in $T$; then the runtime complexity for the primary {\tt smart\_repair} algorithm is at most
$O(E^6)$.
\end{theorem}

In practice, by far the most computationally intense portion of the algorithm is the convexification of the gap sub-boundaries. As for the optional features, the only extra complexity for the region-aware version comes from the addition of the region polygons to the polygons in the tiling $T$, while the rook-to-queen construction is fairly simple.  So Theorem \ref{run-time-theorem} also holds for the {\tt smart\_repair} algorithm with the optional features included.

As for practical performance, here are some sample statistics regarding gaps, overlaps, and the performance of the {\tt Maup} implementation of {\tt smart\_repair} for a few representative state-level precinct maps.  All computations were performed on a 2023 MacBook Pro with M3 Max chip running Python 3.12 and using {\tt Maup} 2.0.3.
\begin{itemize}
\item The 2020 Colorado voting precinct map\footnote{Compiled and provided by the staff of the Colorado Independent Legislative Redistricting Commission} has 3,215 precincts, and the original map contained 909 overlaps and 1,475 gaps.\footnote{Counts of gaps and overlaps were computed using the {\tt doctor} function in the {\tt Maup} package.}  The runtime for {\tt smart\_repair} function on this map was about 10 minutes.
\item The 2023 Wisconsin ward map\footnote{Downloaded from the Wisconsin Legislative Technology Services Bureau at \url{https://gis-ltsb.hub.arcgis.com/}} has 7,013 wards, and the original map contained 4,109 overlaps and 10,821 gaps. The runtime for {\tt smart\_repair} function on this map was slightly under 4 hours.
\item The 2020 New York voting precinct map\footnote{Compiled by the Voting and Elections Science Team (VEST) and downloaded from the Redistricting Data Hub at \url{https://redistrictingdatahub.org}} has 15,376 precincts, and
the original map contained 3,036 overlaps and 7,699 gaps.  The runtime for {\tt smart\_repair} function on this map was slightly under 13 hours.
\end{itemize}
The difference in runtime between the last two examples is somewhat striking, especially since the New York map had many fewer gaps than the Wisconsin map. The longer runtime for the New York map is mainly due to the gaps in the New York map having greater average complexity (as measured by the numbers of edges and sub-boundaries in the gap polygons) than those in the Wisconsin ward map.

\section{Conclusion}\label{back-to-beginning-sec}
We conclude by revisiting the example from Figure \ref{gap_1_maps-fig} that motivated the development of the {\tt smart\_repair} algorithm.
Figure \ref{bad-gap-stretched-fig}(a) shows a close-up view of the gap between the 15 Colorado precincts shown in Figure \ref{gap_1_maps-fig}, disproportionately stretched in the east-west direction so that the shape of the gap is visible.  Figure \ref{bad-gap-stretched-fig}(b)-(c) shows how the gap is filled to repair the precinct polygons using the primary {\tt smart\_repair} algorithm, while Figure \ref{bad-gap-stretched-fig}(d)-(e) shows how the gap is filled to repair the precinct polygons using the county-aware version of the algorithm.  
These images confirm empirically that the adjacency relations between the precincts surrounding this gap in the repaired map are exactly the most natural ones that could be inferred from the original precinct polygons.  Additionally, in the county-aware version we can see that the entire gap lies on the eastern side of the county boundary, as the entire gap is filled by extending precincts on the eastern edge of the gap.
\begin{figure}[h!]
\begin{center}
\subfloat[]{\includegraphics[height=2.2in]{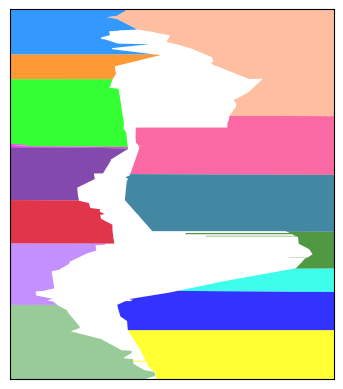} }
\ \ \ 
\subfloat[]{\includegraphics[height=2.2in]{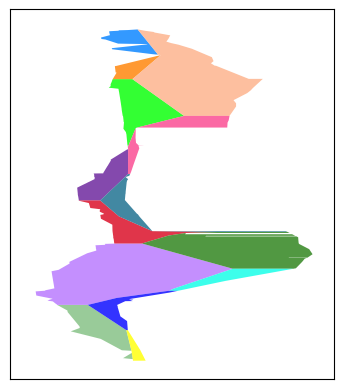} }
\ \ \ 
\subfloat[]{\includegraphics[height=2.2in]{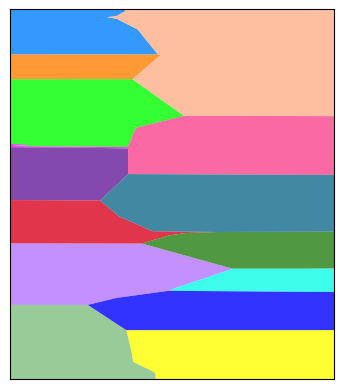} }
\ \ \ 
\subfloat[]{\includegraphics[height=2.2in]{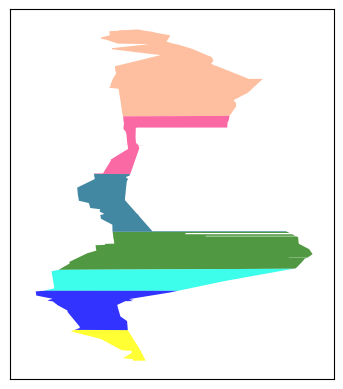} }
\ \ \ 
\subfloat[]{\includegraphics[height=2.2in]{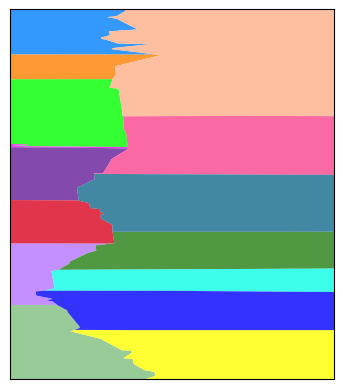} }
\end{center}
\caption{(a) Gap from Figure \ref{gap_1_maps-fig} stretched east-to-west; (b) Gap filled by primary {\tt smart\_repair} algorithm; (c) Precinct polygons repaired by primary {\tt smart\_repair} algorithm; (d) Gap filled by county-aware {\tt smart\_repair} algorithm; (e) Precinct polygons repaired by county-aware {\tt smart\_repair} algorithm}
\label{bad-gap-stretched-fig}
\end{figure}

\section{Data Availability Statement}

The source code for the {\tt smart\_repair} function is available as part of the {\tt Maup} package for Python at \url{https://github.com/mggg/maup}.  The Colorado 2020 precinct shapefile is available from the author on reasonable request.  The Wisconsin 2023 ward shapefile was obtained from  \url{https://gis-ltsb.hub.arcgis.com/} but has since been replaced by a newer version on that website; it is available from the author on reasonable request. The New York 2020 precinct shapefile is available either at 
\url{https://redistrictingdatahub.org} or 
 \url{https://dataverse.harvard.edu/dataverse/electionscience}.

\section{Funding}
This material is based in part upon work supported by the National Science Foundation under Grant No. DMS-1928930 and by the Alfred P. Sloan Foundation under grant G-2021-16778, while the author was in residence at the Simons Laufer Mathematical Sciences Institute (formerly MSRI) in Berkeley, California, during the Fall 2023 semester.

The author was supported in part by a Collaboration Grant for Mathematicians from the Simons Foundation.

\bibliographystyle{amsplain}
\bibliography{smart_repair-bib}

\end{document}